%% file: main.tex
\documentclass{article}

\usepackage{arxiv}

\usepackage[utf8]{inputenc} 
\usepackage[T1]{fontenc}    
\usepackage{hyperref}       
\usepackage{url}            
\usepackage{booktabs}       
\usepackage{amsfonts}       
\usepackage{nicefrac}       
\usepackage{microtype}      
\usepackage{lipsum}		
\usepackage{graphicx}
\usepackage[numbers]{natbib}
\usepackage{doi}
\usepackage{subfigure}
\usepackage{amsmath}
\usepackage{weiwAlgorithm}
\usepackage{balance}
\usepackage{multirow}

\def\BibTeX{{\rm B\kern-.05em{\sc i\kern-.025em b}\kern-.08emT\kern-.1667em\lower.7ex\hbox{E}\kern-.125emX}}

\newcommand{\kc}{$k$-core\xspace}

\newcommand{\ks}{$k$-shell\xspace}

\title{\Large Towards User Engagement Dynamics in Social Networks}

\author{ 
    Qingyuan Linghu \\
    University of New South Wales \\
    \texttt{q.linghu@unsw.edu.au}
	\And
	Fan Zhang \\
	Guangzhou University \\
	\texttt{fanzhang.cs@gmail.com} \\
	\And
	Xuemin Lin \\
	University of New South Wales \\
	\texttt{lxue@cse.unsw.edu.au} \\
	\And
	Wenjie Zhang \\
	University of New South Wales \\
	\texttt{zhangw@cse.unsw.edu.au} \\
	\And
	Ying Zhang \\
	University of Technology Sydney \\
	\texttt{ying.zhang@uts.edu.au} \\
}

\hypersetup{
pdftitle={Towards User Engagement Dynamics in Social Networks},
pdfsubject={cs.SI, cs.DB},
pdfauthor={Qingyuan Linghu, Fan Zhang},
pdfkeywords={Social Network, User Engagement, Graph Algorithm},
}

\begin{document}

\newtheorem{example}{Example}[section]
\newtheorem{definition}{Definition}[section]
\newtheorem{theorem}{Theorem}[section]
\newtheorem{lemma}{Lemma}[section]
\newtheorem{proof}{Proof}[section]

\maketitle

\begin{abstract}
The engagement of each user in a social network is an essential indicator for maintaining a sustainable service. Existing studies use the \textit{coreness} of a user to well estimate its static engagement in a network. However, when the engagement of a user is weakened or strengthened, the influence on other users' engagement is unclear. Besides, the dynamic of user engagement has not been well captured for evolving social networks. In this paper, we systematically study the network dynamic against the engagement change of each user for the first time. The influence of a user is monitored via two novel concepts: the \textit{collapsed power} to measure the effect of user weakening, and the \textit{anchored power} to measure the effect of user strengthening. We show that the two concepts can be naturally integrated such that a unified offline algorithm is proposed to compute both the collapsed and anchored followers for each user. When the network structure evolves, online techniques are designed to maintain the users' followers, which is faster than redoing the offline algorithm by around 3 orders of magnitude. Extensive experiments on real-life data demonstrate the effectiveness of our model and the efficiency of our algorithms.
\end{abstract}


\vspace{2mm}
\section{Introduction}

\begin{figure}[t]
\centering
\subfigure{
\includegraphics[width=0.55\columnwidth]{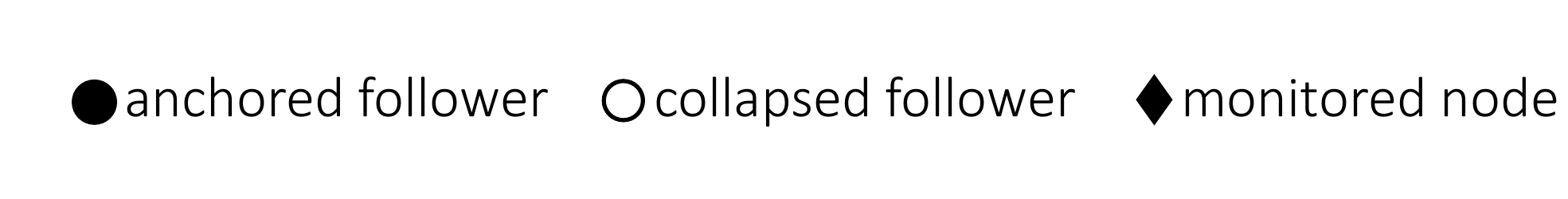}
}\vspace{-2mm}
\subfigure{
\includegraphics[width=0.7\columnwidth]{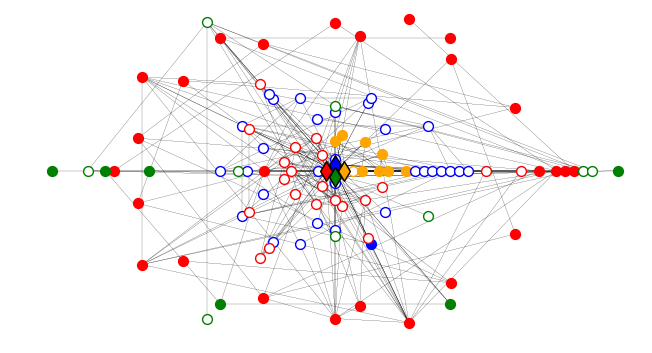}
}
\centering
\caption{Node Monitoring on \texttt{Gowalla}}
\label{fig:introduction}
\end{figure}

In the study of social networks, it is essential to monitor the engagement of users (nodes), and then motivate or protect the critical nodes~\cite{malliaros2013stay,SigmodLinghu}, so as to maintain a sustainable service or to defend the attacks from the competitors.
For instance, Friendster was a once popular social network with over 115 million users, but it is suspended due to contagious leave of users with low engagement~\cite{garcia2013social}~\cite{seki2017mechanism}.

Social networks are often modeled as graphs when studied. In graph theory, the $k$-core is defined as the maximal subgraph where each vertex has at least $k$ neighbors (degree) in the subgraph~\cite{matula1983smallest}~\cite{seidman1983network}. Given
a graph, the $k$-core decomposition iteratively removes every vertex with degree less than $k$. 
Every vertex in the graph has a unique \textit{coreness} value, that is, the largest $k$ s.t.
the $k$-core contains the vertex. Existing works have well studied the effect of coreness on capturing node engagement. In \cite{malliaros2013stay}, the coreness of a node is demonstrated as the "best practice" for its engagement. In \cite{SigmodLinghu}, the engagement of a node (e.g., the check-in number of a node) is further validated as positive correlated with its coreness based on real-life experiments.

The weakening and strengthening of nodes are the two natural engagement dynamics in a network. Specifically, in \cite{collapsed_kcore} (resp. \cite{anchored_kcore}), when weakening (resp. strengthening) a node, we say it is \textit{collapsed} (resp. \textit{anchored}) such that its degree is regarded as 0 (resp. $+\infty$). 
When a node $x$ is collapsed (resp. anchored), all the other nodes which decrease (resp. increase) the coreness values due to collapsing (resp. anchoring) $x$ are called $x$'s \textit{collapsed followers} (resp. \textit{anchored followers}). As far as we know, against the weakening and strengthening of a node, the engagement dynamic change of other nodes has not been systematically studied. Thus, in this paper, we compute the collapsed and anchored followers of each node to monitor its influence on the engagement dynamics of other nodes.

\begin{example}
\label{ex:introduction}
In the social network \texttt{Gowalla}~\cite{snapnets}, we monitor 4 nodes, the \textit{red}, \textit{orange}, \textit{blue} and \textit{green} diamonds at the center of Figure~\ref{fig:introduction}. The collapsed followers (resp. anchored followers) of a monitored node are visualized by the hollow circles (resp. solid circles) with the same color. All the followers having the same coreness are presented at the same distance from the center. We can find the following phenomena: 1) Different nodes have very different number of followers, e.g., the \textit{orange} node has the least followers compared with the other 3 nodes; 2) The ratio of followers of the two categories varies from different monitored nodes, e.g.,  the vast majority followers of the \textit{blue} (resp. \textit{orange}) node are collapsed followers (resp. anchored followers), but the \textit{green} and \textit{red} nodes have relatively even number of collapsed and anchored followers; 3) The coreness of the followers are diverse, e.g., the followers are presented in many different distances from the center even though we only monitor 4 nodes.
\end{example}

From Example~\ref{ex:introduction}, we can find a node's influence  on the engagement dynamics of other nodes is well captured by two aspects: the collapsed power and the anchored power. In this paper, we study the integration and the efficient computation towards anchored and collapsed followers.
As real-life networks may fast evolve, we also aim to efficiently maintain the correct collapsed and anchored followers for each node against edge insertions and deletions, which helps us systematically monitor the node influence on the engagement dynamics of social networks in real-life situation. 

\vspace{2mm}
\section{Related Work}
The stability and engagement dynamic in social networks has attracted significant attention, e.g., \cite{malliaros2013stay}~\cite{anchored_kcore}~\cite{olak}~\cite{chwe2000communication}~\cite{SigmodLinghu}~\cite{dgac}~\cite{collapsed_kcore}~\cite{WWWJZhang}. These works all adopt the $k$-core model because of its degeneration property which can naturally quantify the engagement dynamics of nodes in real-life networks~\cite{matula1983smallest}. The anchored $k$-core~\cite{anchored_kcore} and anchored coreness~\cite{SigmodLinghu} problem aim to find and enhance some key nodes for network stability. Besides, the collasped $k$-core~\cite{collapsed_kcore} and collapsed coreness~\cite{WWWJZhang} problem hold a different view in which some key nodes are protected against the attacks. However, unlike this paper, none of the above works combines collapsing and anchoring as the two natural engagement changes, i.e., the weakening and strengthening, and systematically study the engagement dynamics of other nodes against each node's engagement change in evolving networks.

The concept of $k$-core and its computing algorithm are first introduced in~\cite{matula1983smallest}~\cite{seidman1983network}. An $\mathcal{O}(m)$ in-memory algorithm for core decomposition is provided in~\cite{batagelj2003m}. The state-of-the-art core maintenance algorithm is proposed in~\cite{KOrder}. The core decomposition algorithms under different computation environments are studied including I/O efficient algorithm~\cite{wen2016efficient} and distributed algorithm~\cite{montresor2012distributed}. The $k$-core model has lots of applications including, community discovery~\cite{dourisboure2009extraction}~\cite{fang2017effective}~\cite{li2018skyline}, influential spreader identification~\cite{kitsak2010identification}~\cite{lin2014identifying}~\cite{malliaros2016locating}~\cite{ugander2012structural}, and the applications in biology and ecology~\cite{bader2003automated}~\cite{bola2015dynamic}~\cite{morone2019k}.

Another classic problem that investigates the node influence in networks is \textit{influence maximization} (\texttt{IM}), which is studied in recent years in~\cite{chen2015online}~\cite{li2018influence}~\cite{galhotra2016holistic}~\cite{DBLP:journals/pvldb/GoyalBL11}~\cite{goyal2011celf++}~\cite{lei2015online}~\cite{li2015real}~\cite{DBLP:journals/pvldb/LuCL15}~\cite{ohsaka2014fast}~\cite{DBLP:journals/pvldb/WangFLT17}~\cite{guo2020influence}~\cite{qiu2018deepinf}, to name some. The \texttt{IM} problem is first presented in \cite{DBLP:conf/kdd/KempeKT03} proving the problem is NP-hard. However, the \texttt{IM} problem aims to find $k$ nodes to maximize the social influence instead of monitoring each node's influence regarding network structural stability as in this paper.

\vspace{2mm}
\input{preliminaries}
\vspace{2mm}
\input{algorithms}
\vspace{2mm}
\input{experiments}

\vspace{2mm}
\section{Conclusion}
In this paper, we innovatively propose a model that estimates a node's influence on the engagement dynamics of other nodes in a social network, by two views: the collapsed power and anchored power. Then a parallel algorithm is proposed to compute each node's collapsed and anchored followers offline, equipped with online maintenance technique to update the nodes' followers when the network is evolving. Our effectiveness experiments on real-life datasets demonstrate the necessity of both the collapsed power and anchored power. Our efficiency experiments show that both our offline and online algorithms are efficient.

\vspace{2mm}
\bibliographystyle{abbrv}
\bibliography{reference}

\end{document}

%% file: preliminaries.tex
\section{Preliminaries}

\begin{table}[t]
\small
  \centering
  \caption{Common Notations throughout the paper}
    \label{tb:notations}
    \begin{tabular}{|l|l|}
      \hline
      \textbf{Notation}   & \textbf{Definition}                   \\ \hline \hline

      $G$ & an unweighted and undirected graph  \\ \hline

      $V(G)$; $E(G)$ & the vertex set of $G$; the edge set of $G$  \\ \hline

      $n; m$ & $|V(G)|$; $|E(G)|$ (assume $m > n$) \\ \hline

      $u$, $v$, $w$, $x$ & a vertex in $G$ \\ \hline

      $N(u, G)$ & the set of neighbors of $u$ in $G$  \\ \hline
      
      $deg(u, G)$ & $+\infty$ if $u$ is anchored, 0 if $u$ is collapsed, $|N(u, G)|$ otherwise \\ \hline

      $C_{k}(G)$ & the $k$-core of $G$ \\ \hline

      $c(u, G)$ & the original coreness of $u$ in $G$ \\ \hline
      
      $c^-_x(u, G)$ & the coreness of $u$ in $G$ with collapsing $x$ \\ \hline
      
      $c^+_x(u, G)$ & the coreness of $u$ in $G$ with anchoring $x$ \\ \hline
      
      $^-\mathcal{F}(x, G)$ & the collapsed follower set of $x$ in $G$ \\ \hline
      
      $^+\mathcal{F}(x, G)$ & the anchored follower set of $x$ in $G$ \\ \hline
      
      $H_k(G)$ & the $k$-shell of $G$ \\ \hline
      
      $\mathcal{SC}[v]$ & the only shell component containing $v$ \\ \hline
      
      $S$, $S.V$, $S.E$, $S.c$ & a shell component with its vertex set, edge set and coreness value \\ \hline
      
      $\mathcal{C}[S]$ & the collapser candidate set of $S$\\ \hline
      
      $\mathcal{A}[S]$ & the anchor candidate set of $S$\\ \hline
      
      $^-F[x][S]$ & the collapsed follower set of $x$ in $S$ \\ \hline
      
      $^+F[x][S]$ & the anchored follower set of $x$ in $S$ \\ \hline
     
    \end{tabular}
\end{table}

\label{sec:preli}
We consider an unweighted and undirected graph $G = (V, E)$, where $V(G)$ (resp. $E(G)$) represents the set of vertices (resp. edges) in $G$.
$N(u, G)$ is the set of adjacent vertices of $u$ in $G$, which is also called the neighbour vertex set of $u$ in $G$.
Note that we may omit the input graph for all the notations in the paper when the context is clear, e.g., using $deg(u)$ instead of $deg(u, G)$.

\begin{definition}
\label{def:kcore}
\textbf{$k$-core}.
Given a graph $G$, a subgraph $S$ is the $k$-core of $G$, denoted by $C_{k}(G)$, if ($i$) $S$ satisfies degree constraint, i.e., $deg(u, S) \geq k$ for each $u \in V(S)$;
and ($ii$) $S$ is maximal, i.e., any supergraph $S' \supset S$ is not a $k$-core.
\end{definition}

If $k \ge k'$, the \kc is always a subgraph of $k'$-core, i.e., $C_k(G)\subseteq C_{k'}(G)$. Each vertex in $G$ has a unique coreness.

\begin{definition}
\label{def:corenumber}
\textbf{coreness}.
Given a graph $G$, the coreness of a vertex $u\in V(G)$, denoted by $c(u, G)$, is the largest $k$ such that $C_k(G)$ contains $u$, i.e., $c(u, G) = max\{k~:~u\in C_k(G)\}$.
\end{definition}

\begin{definition}
\label{def:coredecomposition}
\textbf{core decomposition}.
Given a graph $G$, core decomposition of $G$ is to compute the coreness of every vertex in $V(G)$.
\end{definition}

In this paper, once a vertex $x$ in the graph $G$ is \textbf{collapsed} (resp. \textbf{anchored}), the degree of $x$ is regarded as zero (resp. positive infinity), i,e., $deg(x, G) = 0$ (resp. $deg(x, G) = +\infty$).
Every collapsed (resp. anchored) vertex is called an \textbf{collapser} or \textbf{collapser vertex} (resp. \textbf{anchor} or \textbf{anchor vertex}).
The existence of collapser vertices (resp. anchor vertices) may change the coreness of other vertices. We use $c_x^-(u, G)$ (resp. $c_x^+(u, G)$) to denote the coreness of $u$ in $G$ with collapsing (resp. anchoring) $x$.

\begin{algorithm}[t]
\SetVline
\SetFuncSty{textsf}
\SetArgSty{textsf}
\small
\caption{\bf CoreDecomp($G$)}
\label{alg:coredecomp}
\Input{ a graph $G$}
\Output{ $c(u, G)$ for each $u \in V(G)$ }
\State{$k \gets 1$}
\While{exist vertices in $G$}
{
  \While{$\exists u \in V(G)$ with $deg(u) < k$}
  {
    \State{$deg(v) \gets deg(v) - 1$ \textbf{for} each $v \in N(u, G)$}
    \State{remove $u$ and its adjacent edges from $G$}
    \State{$c(u, G) \gets k - 1$}
  }
  \State{$k \gets k + 1$}
}
\Return{$c(u, G)$ for each $u \in V(G)$} 
\end{algorithm}

The computation of core decomposition is shown in Algorithm~\ref{alg:coredecomp}, in which we recursively delete the vertex with the smallest degree in the graph $G$. The time complexity is $\mathcal{O}(m)$.
When there exists a collapser vertex $x$, $x$ is automatically deleted at first as its degree is zero. 
When there exists an anchor vertex $x$, the only difference is that we do not delete $x$ as its degree is positive infinity.

\begin{definition}
\label{def:col_follower}
\textbf{collapsed follower set}.
Given a graph $G$ and the collapser vertex $x$, the collapsed follower set of $x$ in $G$ is denoted by $^-\mathcal{F}(x, G)$, and includes all other vertices decreasing their coreness with collapsing $x$ in $G$, i.e., $^-\mathcal{F}(x, G) = \{u\in V(G):~u \neq x \wedge c^-_x(u, G) < c(u, G)\}$.
\end{definition}

\begin{definition}
\label{def:anc_follower}
\textbf{anchored follower set}.
Given a graph $G$ and the anchor vertex $x$, the anchored follower set of $x$ in $G$ is denoted by $^+\mathcal{F}(x, G)$, and includes all other vertices increasing their coreness with anchoring $x$ in $G$, i.e., $^+\mathcal{F}(x, G) = \{u\in V(G):~u \neq x \wedge c^+_x(u, G) > c(u, G)\}$.
\end{definition}

\noindent \textbf{Problem Statement.}
Given a graph $G$, for each $v \in V(G)$, we compute $^-\mathcal{F}(v, G)$ and $^+\mathcal{F}(v, G)$. When an edge $(u, u')$ is inserted into (resp. removed from) $G$, we have $E(G) \gets E(G) \cup \{(u, u')\}$ (resp. $E(G) \gets E(G) \setminus \{(u, u')\}$). Now for each $v \in V(G)$, we maintain $^-\mathcal{F}(v, G)$ and $^+\mathcal{F}(v, G)$. 

Note that, vertex insertion (resp. removal) can be regarded as a sequence of edges insertion (resp. removal).

%% file: algorithms.tex
\newcommand{\survived}{\textit{survived}\xspace}
\newcommand{\unexplored}{\textit{unexplored}\xspace}
\newcommand{\discarded}{\textit{discarded}\xspace}

\section{The Offline and Online Algorithms}
\label{sec:algorithms}
We introduce our algorithms of offline computing the collapsed follower set (\texttt{CFS}) and anchored follower set (\texttt{AFS}) for each vertex in a graph, and online maintaining the two sets each time when an edge is inserted to or removed from the graph. The computation is based on \textit{shell component} (\ref{subsec:ksc}), by which the graph is partitioned into multiple atom units. Any parallel architecture such as multi-threads can then be utilized to compute the \texttt{CFS} and \texttt{AFS} in parallel. When an edge is inserted or removed, \ref{subsec:maintenance} shows how we maintain the atom units. We prove that, by only conducting computation regarding those new updated atom units, the \texttt{CFS} and \texttt{AFS} for each vertex can be efficiently maintained. The further accelerating techniques for computing \texttt{CFS} and \texttt{AFS} are introduced in \ref{subsec:followers}.

\subsection{The Shell Component Based Framework}
\label{subsec:ksc}

\vspace{1mm}
\begin{definition}
\label{def:k_shell}
\textbf{$k$-shell}.
Given a graph $G$, the $k$-shell, denoted by $H_{k}(G)$, is the set of vertices in $G$ with coreness equal to $k$, i.e., $H_{k}(G) = V(C_k(G)) \setminus V(C_{k+1}(G))$.
\end{definition}

\begin{figure}[htp]
\centerline{\includegraphics[width=0.4\columnwidth]{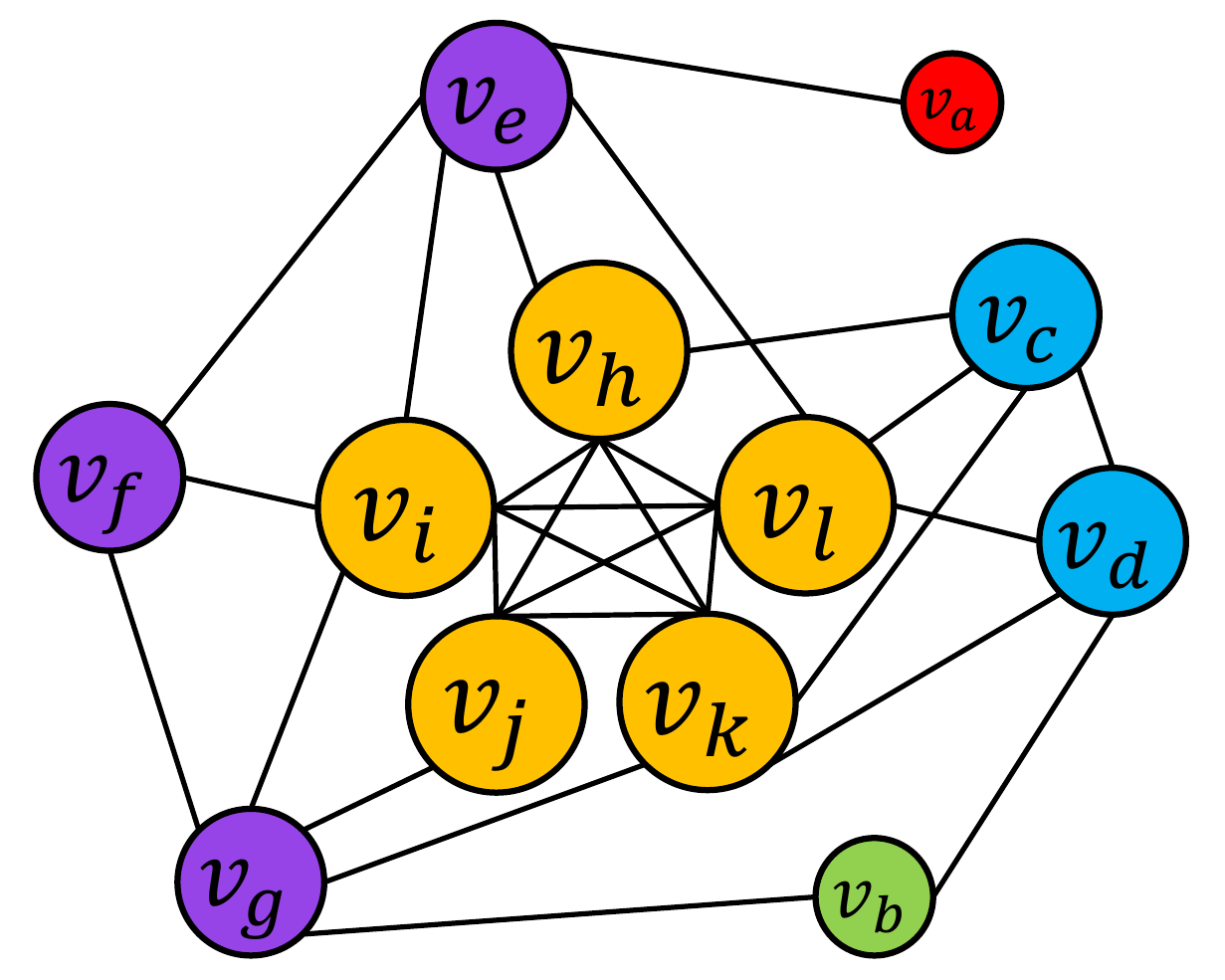}}
\caption{Illustration of Shell Component.}
\label{fig:shell_cp}
\end{figure}

\begin{definition}
\label{def:shellcomponent}
\textbf{shell component}.
Given a graph $G$ and the \ks $H_k(G)$, a subgraph $S$ is a shell component of $H_k(G)$, if $S$ is a maximal induced connected component of $H_k(G)$.
\end{definition}

\begin{theorem}
\label{theo:only_sc}
For each vertex $v \in V(G)$, there exists only one shell component $S$ having $v \in V(S)$.
\end{theorem}
\begin{proof}
We prove it by contradiction. Assuming there are $S_1$ and $S_2$ with $v \in V(S_1)$ and $v \in V(S_2)$. Then we have, for each $u \in V(S_1)$, $c(u) = c(v)$ and $u$ is connected with $v$; For each $u \in V(S_2)$, $c(u) = c(v)$ and $u$ is connected with $v$. Thus, $S_1 \cup S_2$ satisfies the definition of one single shell component, which contradicts with our assumption. Proof completes.
\end{proof}

\begin{example}
\label{ex:ks}
In Figure~\ref{fig:shell_cp}, we have $H_1(G) = \{v_a\}$, $H_2(G) = \{v_b\}$, $H_3(G) = \{v_c, v_d, v_e, v_f, v_g\}$ and $H_4(G) = \{v_h, v_i, v_j, v_k, v_l\}$. For $H_1(G)$, $S_1$ is the only shell component of $H_1(G)$ with $V(S_1) = \{v_a\}$. For $H_2(G)$, $S_2$ is the only shell component with $V(S_2) = \{v_b\}$. But within $H_3(G)$, we have two shell components $S_3$ and $S_4$, in which $V(S_3) = \{v_c, v_d\}$ and $V(S_4) = \{v_e, v_f, v_g\}$, because $S_3$ and $S_4$ are not connected by the edges among $H_3(G)$. For $H_4(G)$, $S_5$ is the only shell component with $V(S_5) = \{v_h, v_i, v_j, v_k, v_l\}$. The graph in Figure~\ref{fig:shell_cp} will be often referred in the following examples, so we summarize the notations in Table~\ref{tb:shell_component} for your convenience.
\end{example}

\begin{table}[htb]
\small
\centering
\caption{\small Shell Components in Figure~\ref{fig:shell_cp}}
\label{tb:shell_component}
\begin{tabular}{|l|l|l|l|l|l|}
\hline
\textit{Shell component} & $S_1$ & $S_2$ & $S_3$ & $S_4$ & $S_5$ \\ \hline
\textit{Vertex set} & $v_a$ & $v_b$ & $v_c$, $v_d$ & $v_e$, $v_f$, $v_g$ & $v_h$, $v_i$, $v_j$, $v_k$, $v_l$ \\ \hline
\end{tabular}
\end{table}

For a shell component $S$ of $H_k(G)$, we denote $S.V$, $S.E$ and $S.c$ as the vertex set, edge set and the coreness of the vertices in $S$, i.e., $S.V = V(S)$, $S.E = E(S)$ and $S.c = c(v, G)~\forall v \in S.V$. We use the structure $\mathcal{SC}$ to index the shell components for all the vertices, in which for each $v \in V(G)$, $\mathcal{SC}[v]$ is the only shell component having $v \in \mathcal{SC}[v].V$ (Theorem~\ref{theo:only_sc}).
Algorithm~\ref{alg:sdecomp} and Algorithm~\ref{alg:shellconn} illustrate the process of decomposing each vertex into its shell component.

\begin{algorithm}[t]
\SetVline
\SetFuncSty{textsf}
\SetArgSty{textsf}
\small
\caption{\bf ShellDecomp($G$)}
\label{alg:sdecomp}

\Input
{
    $G:$ the graph
}
\Output
{
    $\mathcal{SC}:$ the index of shell components in $G$
}

\State{\textbf{CoreDecomp}($G$)}
\For{each $unassigned$ $u\in V(G)$}
{
    \State{$S \gets$ a new shell component}
    \State{$S.c := c(u, G)$}
    \State{$S.V := S.V\cup \{u\}$}
    \State{$u$ is set $assigned$}
    \State{\textbf{ShellConnect}($u$, $S$, $\mathcal{SC}$)}
    \State{$\mathcal{SC}[u] := S$}
}

\Return{$\mathcal{SC}$}
\end{algorithm}

\begin{algorithm}[t]
\SetVline
\SetFuncSty{textsf}
\SetArgSty{textsf}
\small
\caption{\bf ShellConnect($u$, $S$, $\mathcal{SC}$)}
\label{alg:shellconn}

\Input
{
    $u:$ a vertex, $S:$ the shell component containing $u$, $\mathcal{SC}:$ the shell component index
}

\For{each $v\in N(u, G)$}
{
    \If{$c(v) = c(u)$}
    {
        \State{$S.E := S.E\cup \{(u, v)\}$}
        \If{$v$ is $unassigned$}
        {
            \State{$S.V := S.V \cup \{v\}$}
            \State{$v$ is set $assigned$}
            \State{\textbf{ShellConnect}($v$, $S$, $\mathcal{SC}$)}
            \State{$\mathcal{SC}[v] := S$}
        }
    }
}
\end{algorithm}

In Algorithm~\ref{alg:sdecomp}, firstly we need to conduct core decomposition (Line 1) on $G$ so that we can get the coreness of each vertex. We traverse all the vertices with each vertex marked $unassigned$ as default (Line 2). Each time meeting an \textit{unassigned} vertex $u \in V(G)$, we create a new shell component $S$ (Line 3), set the related domains of $S$ and set $u$ as \textit{assigned} (Line 4-6). Then we call Algorithm~\ref{alg:shellconn} (details in the next paragraph) to recursively collect all the vertices which should be in $S$ (Line 7), followed by assigning $\mathcal{SC}[u] := S$ (Line 8). When all the vertices are set $assigned$ (in Algorithm~\ref{alg:sdecomp} or Algorithm~\ref{alg:shellconn}), we get the complete $\mathcal{SC}$. The time complexity of Algorithm~\ref{alg:sdecomp} is $\mathcal{O}(m)$ as we traverse each vertex's neighbors once.

In Algorithm~\ref{alg:shellconn}, for the vertex $u$, we traverse all its neighbors in $N(u, G)$ (Line 1). If $c(v) = c(u)$ (Line 2), we add the edge $(u, v)$ into $S.E$ (Line 3). Note that $(u, v)$ and $(v, u)$ are the same in our setting. Only if $v$ is \textit{unassigned} (Line 4), we add $v$ into $S.V$ and recursively call Algorithm~\ref{alg:shellconn} to find all the vertices of $S$ (Line 5-8).

\vspace{1mm}
\textbf{$\mathcal{C}[\cdot]$ and $\mathcal{A}[\cdot]$.} We define two affiliated structures, the collapser candidate set $\mathcal{C}[\cdot]$ and the anchor candidate set $\mathcal{A}[\cdot]$ w.r.t. all the shell components of $G$. For a shell component $S$, $\mathcal{C}[S] = \{v:~v \in S.V \vee c(v, G) > S.c \wedge N(v, G) \cap S.V \neq \emptyset\}$, and $\mathcal{A}[S] = \{v:~v \in S.V \vee c(v, G) < S.c \wedge N(v, G) \cap S.V \neq \emptyset\}$. $\mathcal{C}[\cdot]$ and $\mathcal{A}[\cdot]$ can be easily computed when traversing each vertex's neighbors in Algorithm~\ref{alg:sdecomp} or Algorithm~\ref{alg:shellconn}, without changing the time complexity.

\begin{lemma}
\label{lemma:col_one_dec_one}
If a vertex $x$ is collapsed in $G$, any other vertex $u \in V(G) \setminus \{x\}$ decreases its coreness by at most 1.
\end{lemma}
\begin{proof}
Please refer to the proof of Theorem 4 in~\cite{WWWJZhang}.
\end{proof}

\vspace{1mm}
\begin{lemma}
\label{lemma:anc_one_inc_one}
If a vertex $x$ is anchored in $G$, any other vertex $u \in V(G) \setminus \{x\}$ increases its coreness by at most 1.
\end{lemma}
\begin{proof}
Please refer to the proof of Theorem 4.6 in~\cite{SigmodLinghu}.
\end{proof}

\vspace{1mm}
\begin{theorem}
\label{theo:c_index}
For collapsing $x$ in $G$, another vertex $u$ is $x$'s
collapsed follower indicates $x \in \mathcal{C}[S]$ s.t. $u \in S.V$.
\end{theorem}
\begin{proof}
We prove it by contradiction. Assuming $u$ decreases its coreness since $x$ is collapsed, $u \in S.V$ and $x \notin \mathcal{C}[S]$. This firstly means $x \notin S.V$, according to the definition of $\mathcal{C}[S]$, then we have the following possible situations: 1) $c(x, G) < S.c$; 2) $c(x, G) = S.c$; 3) $c(x, G) > S.c \wedge N(x, G) \cap S.V = \emptyset$.

For situation 1), No matter $x$ is collapsed or not, $x$ is always deleted before $u$ in core decomposition (Algorithm~\ref{alg:coredecomp}), so that $c^-_x(u) = c(u)$ which contradicts with our assumption.

For situation 2), Consider deleting $V(G)$ in core decomposition without collapsing $x$. After all the vertices with coreness less than $S.c$ are deleted, for each $S' \neq S$ s.t. $S'.c = S.c$, it is available to delete $S'.V$ before $S.V$. And $S.V$ can remain in the graph still satisfying for each $v \in S.V$, $deg(v) \geq S.c$. We denote the deleted (resp. remained) vertex set so far as $V^d$ (resp. $V^r$). For each $v \in V^r \setminus S.V$, $c(v) > S.c$. Because $c(x) = S.c$ and $x \notin S.V$, we can conclude $x \in V^d$. For the graph with collapsing $x$ (deleted firstly), let us then delete $V^d \setminus \{x\}$. Now the remaining graph are still induced by $V^r$ satisfying for each $v \in V^r$, $deg(v) \geq S.c$. Since $u \in V^r$, $u$ is not a collapsed follower of $x$, which contradicts with our assumption.

For situation 3), in graph $G$ without collapsing $x$, for each $v \in S.V$, we denote $N^>(v) = \{w \in N(v, G):~c(w) > S.c\}$. As $c(v) = S.c$, we have $|N(v, S) \cup N^>(v)| \geq S.c$. With collapsing $x$ in $G$, core decomposition firstly deletes each $v \in V(G)$ s.t. $c^-_x(v) < S.c$. Reconsider each $w \in N^>(v)$ s.t. each $v \in S.V$. For situation 3), $w \neq x$. And based on Lemma~\ref{lemma:col_one_dec_one}, $c^-_x(w) \geq S.c$. Thus, $w$ remains in the graph. Now consider each $v \in S.V$, $N^>(v)$ are complete in the remaining graph, so we still have $|N(v, S) \cup N^>(v)| \geq S.c$. Thus, $v$ is not a collapsed follower of $x$. As $u \in S.V$, this contradicts with our assumption.
\end{proof}

\vspace{1mm}
\begin{theorem}
\label{theo:a_index}
For anchoring $x$ in $G$, another vertex $u$ is $x$'s
anchored follower indicates $x \in \mathcal{A}[S]$ s.t. $u \in S.V$.
\end{theorem}
\begin{proof}
We prove it by contradiction. Assuming $u$ increases its coreness since $x$ is anchored, $u \in S.V$ and $x \notin \mathcal{A}[S]$. This firstly means $x \notin S.V$, according to the definition of $\mathcal{A}[S]$, then we have the following possible situations: 1) $c(x, G) > S.c$; 2) $c(x, G) = S.c$; 3) $c(x, G) < S.c \wedge N(x, G) \cap S.V = \emptyset$.

For situation 1), No matter $x$ is anchored or not, $x$ is always deleted after $u$ in core decomposition (Algorithm~\ref{alg:coredecomp}), so that $c^+_x(u) = c(u)$ which contradicts with our assumption.

For situation 2), Consider deleting $V(G)$ in core decomposition without anchoring $x$. After all the vertices with coreness less than $S.c$ are deleted, it is available to delete $S.V$ before each $S'.V$ s.t. $S' \neq S$ and $S'.c = S.c$. And each such $S'$ can remain in the graph still satisfying for each $v \in S'.V$, $deg(v) \geq S.c$. We denote the deleted vertex set so far as $V^d$. Because $c(x) = S.c$ and $x \notin S.V$, we can conclude $x \notin V^d$. For the graph with anchoring $x$, as $x \notin V^d$, we can safely still delete $V^d$ following the same vertex sequence. Thus for each $v \in V^d$, $c^+_x(v) = c(v)$. Since $u \in V^d$, $u$ is not an anchored follower of $x$, which contradicts with our assumption.

For situation 3), in graph $G$ without anchoring $x$, for each $v \in S.V$, we denote $N^<(v) = \{w \in N(v, G):~c(w) < S.c\}$ and $N^>(v) = \{w \in N(v, G):~c(w) > S.c\}$. As $c(v) = S.c$, we have $|N(v, S) \cup N^>(v)| < S.c + 1$. When conducting core decomposition with anchoring $x$ in $G$, reconsider each $w \in N^<(v)$ s.t. each $v \in S.V$. For situation 3), $w \neq x$. And based on Lemma~\ref{lemma:anc_one_inc_one}, $c^+_x(w) < S.c + 1$, thus each such $w$ is not in the $(S.c+1)$-core. Now consider each $v \in S.V$, $N^<(v)$ are not in the $(S.c+1)$-core and we still have $|N(v, S) \cup N^>(v)| < S.c + 1$. Thus, $v$ is not an anchored follower of $x$. As $u \in S.V$, this contradicts with our assumption.
\end{proof}

\vspace{1mm}
\textbf{Parallel Computation.} For each $v \in V(G)$, We define $^-F[v][S] = \{u \in S.V:~u \in {^-\mathcal{F}(v, G)}\}$ and $^+F[v][S] = \{u \in S.V:~u \in {^+\mathcal{F}(v, G)}\}$, so as to compute $v$'s collapsed followers and anchored followers in each atom unit, i.e., the shell component. Based on Theorems \ref{theo:c_index} and \ref{theo:a_index}, we know that for the vertex $v$, its collapsed followers (resp. anchored followers) are only from each $S.V$ s.t. $v \in \mathcal{C}[S]$ (resp. $v \in \mathcal{A}[S]$). Thus, in Algorithm~\ref{alg:parallel_comp}, we parallelly compute $^-F[v][S]$ (resp. $^+F[v][S]$) w.r.t. $\{S:~v \in \mathcal{C}[S]\}$ (resp. $\{S:~v \in \mathcal{A}[S]\}$). Then $^-\mathcal{F}(v, G) = \bigcup_{\{S:~v \in \mathcal{C}[S]\}}{^-F[v][S]}$ and $^+\mathcal{F}(v, G) = \bigcup_{\{S:~v \in \mathcal{A}[S]\}}{^+F[v][S]}$. In order to unify the offline computation and the online maintenance when any edge is inserted or removed, Algorithm~\ref{alg:parallel_comp} takes a new shell component set $\hat{S}$ as an input. Now we let $\hat{S}$ contain all the shell components in $G$, i.e., $\hat{S} := \bigcup_{u \in V(G)}\{\mathcal{SC}[u]\}$, then call \texttt{ParallelFollowerComp}($\mathcal{C}[\cdot]$, $\mathcal{A}[\cdot]$, $\hat{S}$). For each $S \in \hat{S}$, whenever there exists an available thread (Line 1), the computation in Line 2-7 can be conducted, followed by releasing this thread. As the number of shell components in $G$ is much more than the number of available threads, the dynamic scheduling ensures the computing resources are mostly utilized in parallel. Based on theorem~\ref{theo:only_sc}, each $^-F[v][S]$ (resp. $^+F[v][S]$) does not overlap with others, which means the threads are never wasted on redundant computations. The functions \texttt{FindCollapsedFollowers} (Line 3) and \texttt{FindAnchoredFollowers} (Line 6) are equipped with further accelerating techniques, which are presented by Algorithm~\ref{alg:col_followers} and Algorithm~\ref{alg:anc_followers} in subsection \ref{subsec:followers}, respectively.

\begin{figure}[htp]
\centerline{\includegraphics[width=0.7\columnwidth]{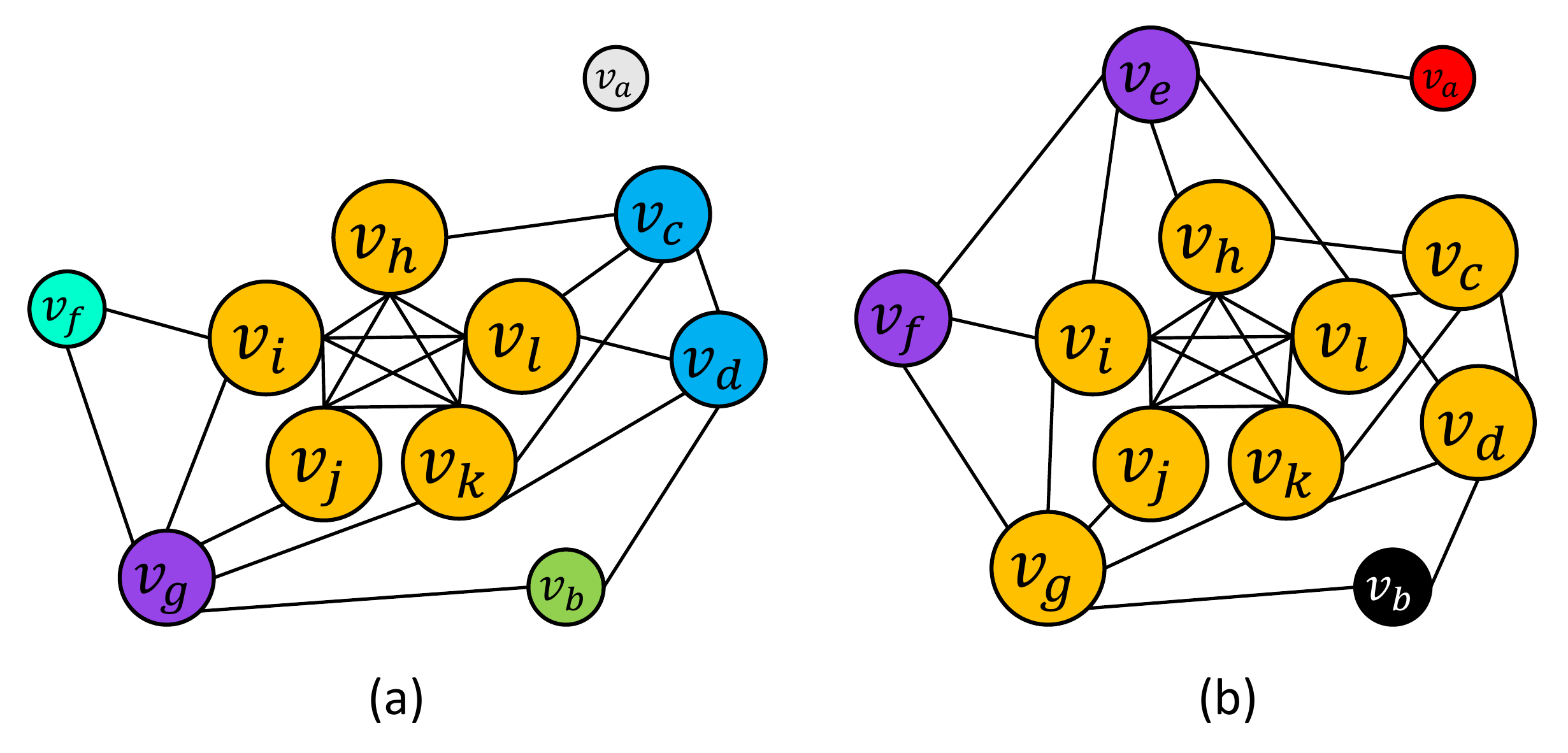}}
\caption{Followers Computation.}
\label{fig:followers}
\end{figure}

\begin{example}
\label{ex:followers_col}
For the graph in Figure~\ref{fig:shell_cp}, Figure~\ref{fig:followers} (a) shows the graph with collapsing $v_e$. We retain the notations in Example~\ref{ex:ks}. Please refer to Table~\ref{tb:shell_component} for your convenience. There are two shell components $S_1$ and $S_4$ such that $v_e \in \mathcal{C}[S_1]$ and $v_e \in \mathcal{C}[S_4]$. Then $^-F[v_e][S_1]$ and $^-F[v_e][S_4]$ can be computed in parallel. We have $^-F[v_e][S_1] = \{v_a\}$ with $c^-_{v_e}(v_a) = 0$, and $^-F[v_e][S_4] = \{v_f\}$ with $c^-_{v_e}(v_f) = 2$.
\end{example}

\begin{example}
\label{ex:followers_anc}
For the graph in Figure~\ref{fig:shell_cp}, Figure~\ref{fig:followers} (b) shows the graph with anchoring $v_b$. We retain the notations in Example~\ref{ex:ks}. Please refer to Table~\ref{tb:shell_component} for your convenience. There are three shell components $S_2$, $S_3$ and $S_4$ such that $v_b \in \mathcal{A}[S_2]$, $v_b \in \mathcal{A}[S_3]$ and $v_b \in \mathcal{A}[S_4]$. $^+F[v_b][S_2]$, $^+F[v_b][S_3]$ and $^+F[v_b][S_4]$ can be computed in parallel. We have $^+F[v_b][S_2] = \emptyset$, $^+F[v_b][S_3] = \{v_c, v_d\}$ with $c^+_{v_b}(v_c) = c^+_{v_b}(v_d) = 4$, $^+F[v_b][S_4] = \{v_g\}$ with $c^+_{v_b}(v_g) = 4$.
\end{example}

\begin{algorithm}[t]
\SetVline 
\SetFuncSty{textsf}
\SetArgSty{textsf}
\small
\caption{\bf ParallelFollowerComp($\mathcal{C}[\cdot]$, $\mathcal{A}[\cdot]$, $\hat{S}$)}
\label{alg:parallel_comp}
\Input
{
    $\mathcal{C}[\cdot]:$ the collapsed candidate set, $\mathcal{A}[\cdot]:$ the anchored candidate set, $\hat{S}:$ the new shell component set
}
\Output
{
    $^-\mathcal{F}(v, G)$ and $^+\mathcal{F}(v, G)$ for each $v \in V(G)$
}
\For{each $S \in \hat{S}$ in \texttt{\underline{dynamic multithreads}}}
{
    \For{each $v \in \mathcal{C}[S]$}
    {
        \State{$^-F[v][S] := \textbf{FindCollapsedFollowers}(v, S)$}
        \State{$^-\mathcal{F}(v, G) := {^-\mathcal{F}(v, G)} \cup {^-F[v][S]}$}
    }
    \For{each $v \in \mathcal{A}[S]$}
    {
        \State{$^+F[v][S] := \textbf{FindAnchoredFollowers}(v, S)$}
        \State{$^+\mathcal{F}(v, G) := {^+\mathcal{F}(v, G)} \cup {^+F[v][S]}$}
    }
}
\Return{$^-\mathcal{F}(v, G)$ and $^+\mathcal{F}(v, G)$ for each $v \in V(G)$}
\end{algorithm}

\subsection{The Maintenance w.r.t. Edge Streaming}
\label{subsec:maintenance}

We consider the following two situations of edge streaming: 1) a single edge is removed from the graph; 2) a single edge is inserted into the graph. Please note that, multiple edges and vertices streaming can be regarded as a sequence of situations 1) and 2). We maintain the collapsed follower set $^-\mathcal{F}(v, G)$ and anchored follower set $^+\mathcal{F}(v, G)$ for each $v \in V(G)$ by Algorithm~\ref{alg:smaintenance} which unifies situations 1) and 2). 
Specifically, in the updated graph $G$ ($E(G) \cup \{(v_s, v_t)\}$ or $E(G) \setminus \{(v_s, v_t)\}$), we firstly adopt the state-of-the-art algorithm of core maintenance in \cite{KOrder} to update $c(u, G)$ for each $u \in V(G)$ (Line 1), followed by resetting each vertex as \textit{unassigned} (Line 2-3). In Line 4, $\mathcal{SC'}[\cdot]$, $\mathcal{C'}[\cdot]$ and $\mathcal{A'}[\cdot]$ are initially copied from the current $\mathcal{SC}[\cdot]$, $\mathcal{C}[\cdot]$ and $\mathcal{A}[\cdot]$, and they will be properly updated. Line 5-13 collects the new shell components into $\hat{S}$. Without the need to call Algorithm~\ref{alg:sdecomp} again for the whole graph, the maintenance process only needs to start from each vertex $u$ in $\mathcal{SC}[v_s]$ and $\mathcal{SC}[v_t]$ (Line 6). Line 7-12 do the same operations as Line 3-8 of Algorithm~\ref{alg:sdecomp}. Please note that $\mathcal{C}'[\cdot]$ and $\mathcal{A}'[\cdot]$ can be updated straightforward during Line 6-13 according to their definitions (Line 14), because only $\mathcal{A}'[S']$ and $\mathcal{C}'[S']$ for each $S' \in \hat{S}$ need to be updated. We denote an updated vertex set as $U = \bigcup_{S' \in \hat{S}} S'.V$ (Line 15). By traversing all the old shell components of the vertices in $U$ (using $\mathcal{SC}$ instead of $\mathcal{SC}'$ in Line 16), we can remove the expired collapsed followers and anchored followers (Line 17-20, using $\mathcal{C}[\cdot]$ and $\mathcal{A}[\cdot]$ instead of $\mathcal{C'}[\cdot]$ and $\mathcal{A'}[\cdot]$). At last, for the new shell component set $\hat{S}$, we call Algorithm~\ref{alg:parallel_comp} with inputting the updated $\mathcal{C'}[\cdot]$ and $\mathcal{A'}[\cdot]$ (Line 21) to compute the new followers.

\begin{algorithm}[t]
\SetVline
\SetFuncSty{textsf}
\SetArgSty{textsf}
\small
\caption{\bf FollowerMaintain($(v_s, v_t)$, $G$)}
\label{alg:smaintenance}

\Input
{
    $(v_s, v_t):$ the inserted or removed edge, $G:$ the graph with $E(G) \cup \{(v_s, v_t)\}$ or $E(G) \setminus \{(v_s, v_t)\}$
}
\State{Conduct the core maintenance~\cite{KOrder} to get $c(u, G)$ for each $u \in V(G)$}
\For{each $u \in V(G)$}
{
    \State{$u$ is set as $unassigned$}
}
\State{$\mathcal{SC'}[\cdot]$, $\mathcal{C'}[\cdot]$, $\mathcal{A'}[\cdot]$ $\gets$ $\mathcal{SC}[\cdot]$, $\mathcal{C}[\cdot]$, $\mathcal{A}[\cdot]$}
\State{$\hat{S} \gets$ the new shell component set}
\For{each $unassigned$ $u \in S.V$ s.t. each $S \in \{\mathcal{SC}[v_s], \mathcal{SC}[v_t]\}$}
{
    \State{$S' \gets$ a new shell component}
    \State{$S'.c := c(u, G)$}
    \State{$S'.V := S'.V\cup \{u\}$}
    \State{$u$ is set $assigned$}
    \State{\textbf{ShellConnect}($u$, $S'$, $\mathcal{SC'}$)}
    \State{$\mathcal{SC'}[u] := S'$}
    \State{$\hat{S} = \hat{S} \cup \{S'\}$}
}
\State{$\mathcal{C'}[\cdot]$ and $\mathcal{A'}[\cdot]$ are updated while doing Line 6-13}
\State{$U \gets \bigcup_{S' \in \hat{S}} S'.V$}
\For{each $S \in \bigcup_{u \in U} \{\mathcal{SC}[u]\}$}
{
    \For{each $v \in \mathcal{C}[S]$}
    {
        \State{$^-\mathcal{F}(v, G) := {^-\mathcal{F}(v, G)} \setminus {^-F[v][S]}$}
    }
    \For{each $v \in \mathcal{A}[S]$}
    {
        \State{$^+\mathcal{F}(v, G) := {^+\mathcal{F}(v, G)} \setminus {^+F[v][S]}$}
    }
}
\State{\textbf{ParallelFollowerComp}($\mathcal{C'}[\cdot]$, $\mathcal{A'}[\cdot]$, $\hat{S}$)}
\end{algorithm}

\begin{figure}[htp]
\centerline{\includegraphics[width=0.7\columnwidth]{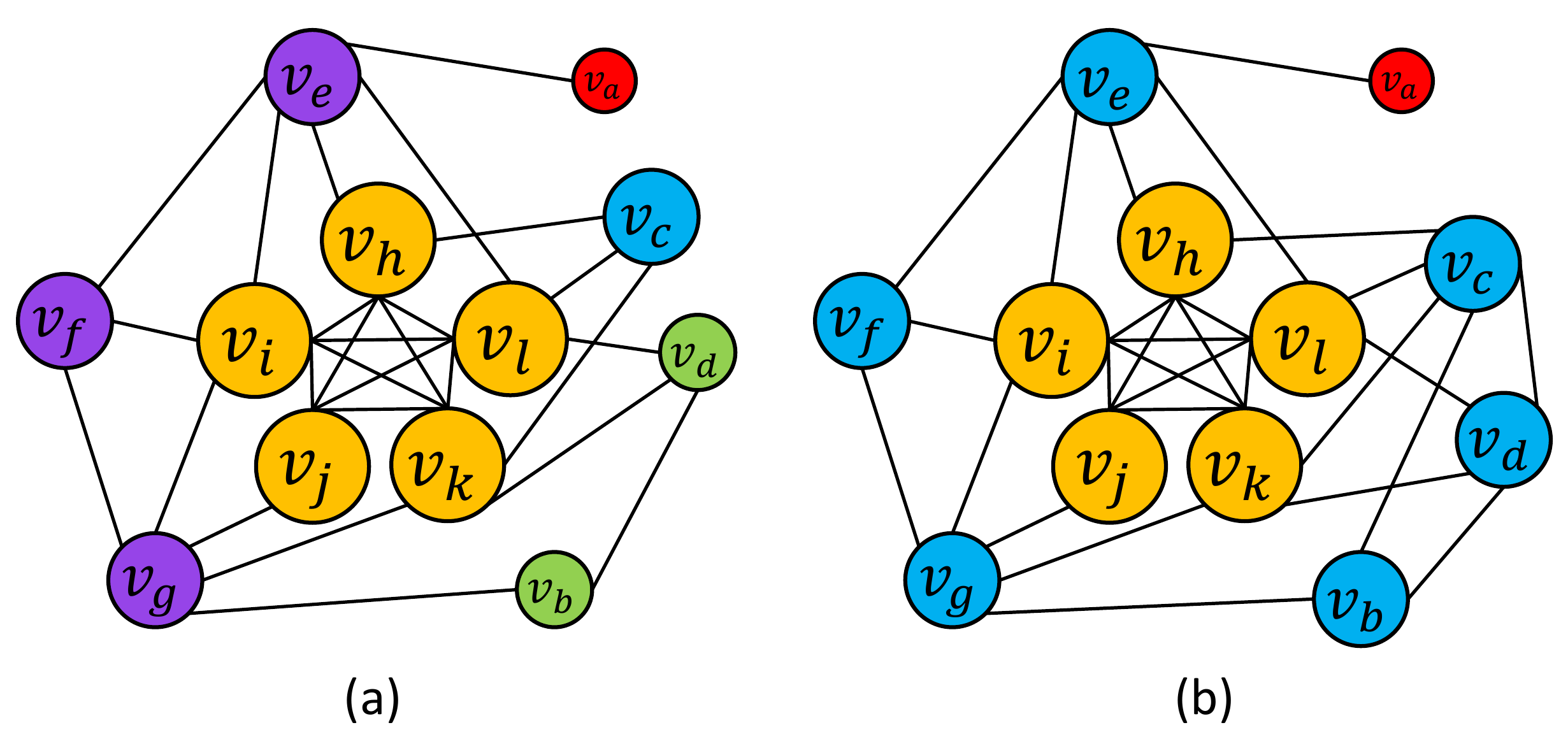}}
\caption{Followers Maintenance.}
\label{fig:maintenance}
\end{figure}

\begin{table}[htb]
\small
\centering
\caption{\small Followers Maintenance w.r.t. removing $(v_c, v_d)$}
\label{tb:maintenance_rmv}
\begin{tabular}{|p{0.03\columnwidth}|p{0.3\columnwidth}|p{0.3\columnwidth}|}
\hline
 & $\mathcal{C}'[\cdot]$ & $\mathcal{A}'[\cdot]$ \\ \hline
$S_1$ & $v_a$, $v_e$ & $v_a$ \\ \hline
$S_4$ & $v_e$, $v_f$, $v_g$, $v_h$, $v_i$, $v_j$, $v_k$, $v_l$ & $v_a$, $v_e$, $v_f$, $v_g$ \\ \hline
$S_5$ & $v_h$, $v_i$, $v_j$, $v_k$, $v_l$ & $v_c$, $v_d$, $v_e$, $v_f$, $v_g$, $v_h$, $v_i$, $v_j$, $v_k$, $v_l$ \\ \hline \hline
$S_6$ & $v_b$, $v_d$, $v_g$, $v_k$, $v_l$ & $v_b$, $v_d$ \\ \hline
$S_7$ & $v_c$, $v_h$, $v_k$, $v_l$ & $v_c$ \\ \hline
\end{tabular}
\end{table}

\begin{table}[htb]
\small
\centering
\caption{\small Followers Maintenance w.r.t. inserting $(v_b, v_c)$}
\label{tb:maintenance_ist}
\begin{tabular}{|p{0.03\columnwidth}|p{0.3\columnwidth}|p{0.3\columnwidth}|}
\hline
 & $\mathcal{C}'[\cdot]$ & $\mathcal{A}'[\cdot]$ \\ \hline
$S_1$ & $v_a$, $v_e$ & $v_a$ \\ \hline
$S_5$ & $v_h$, $v_i$, $v_j$, $v_k$, $v_l$ & $v_c$, $v_d$, $v_e$, $v_f$, $v_g$, $v_h$, $v_i$, $v_j$, $v_k$, $v_l$ \\ \hline \hline
$S_8$ & $v_b$, $v_c$, $v_d$, $v_e$, $v_f$, $v_g$, $v_h$, $v_i$, $v_j$, $v_k$, $v_l$ & $v_a$, $v_b$, $v_c$, $v_d$, $v_e$, $v_f$, $v_g$ \\ \hline
\end{tabular}
\end{table}

\begin{example}
\label{ex:maintenance_rmv}
In Figure~\ref{fig:maintenance} (a), we remove the edge $(v_c, v_d)$ from the graph in Figure~\ref{fig:shell_cp}. All the notations in Example~\ref{ex:ks} (Table~\ref{tb:shell_component}) are retained. $S_2$ and $S_3$ are expired. We have the new $S_6$ with $S_6.V = \{v_b, v_d\}$ and $S_6.c = 2$, and the new $S_7$ with $S_7.V = \{v_c\}$ and $S_7.c = 3$. Table~\ref{tb:maintenance_rmv} shows all the updated shell components and the vertices in $\mathcal{C}'[\cdot]$ and $\mathcal{A}'[\cdot]$. The collapsed and anchored followers of $S_1$, $S_4$ and $S_5$ retain the same (row 2-4). Only the collapsed and anchored followers of $S_6$ and $S_7$ need to be computed (row 5-6).
\end{example}

\begin{example}
\label{ex:maintenance_ist}
In Figure~\ref{fig:maintenance} (b), we insert the edge $(v_b, v_c)$ into the graph in Figure~\ref{fig:shell_cp}. All the notations in Table~\ref{tb:shell_component} are retained. $S_2$, $S_3$ and $S_4$ are expired. We have the new $S_8$ with $S_8.V = \{v_b, v_c, v_d, v_e, v_f, v_g\}$ and $S_8.c = 3$. Table~\ref{tb:maintenance_ist} shows all the updated shell components and the vertices in $\mathcal{C}'[\cdot]$ and $\mathcal{A}'[\cdot]$. The collapsed and anchored followers of $S_1$ and $S_5$ retain the same (row 2-3). Only the collapsed and anchored followers of $S_8$ need to be computed (row 4).
\end{example}

\subsection{The Efficient Followers Computation}
\label{subsec:followers}

\begin{algorithm}[t]
\SetVline 
\SetFuncSty{textsf}
\SetArgSty{textsf}
\small
\caption{\bf FindCollapsedFollowers($x$, $S$)}
\label{alg:col_followers}
\Input
{
  $x:$ the collapser candidate, $S:$ a shell component
}
\Output
{
  $^-F[x][S]:$ the collapsed follower set of $x$ in $S$
}
\State{$Q \gets$ a queue}
\If{$x \in S.V$}
{
    \State{$x$ is set \discarded; $Q.push(x)$}
}
\Else
{
    \For{each $u \in N(x, G) \cap S.V$}
    {
        \State{$HS(u) := HS(u) - 1$; $Q.push(u)$}
    }
}
\While{$Q \not = \emptyset$ }
{
    \State{$u \leftarrow Q.pop()$}
    \If{$u \neq x$}
    {
        \State{$d^+(u) := HS(u) + |\{v:~v\in N(u, S) \wedge v~is~not~\discarded\}|$}
        \If{$d^+(u) < S.c$}
        {
            \State{$u$ is set \discarded}
        }
    }
    \If{$u$ is \discarded}
    {
        \For{each $v \in N(u, S)$ with $v$ is not \discarded and $v \notin Q$}
        {
            \State{$Q.push(v)$}
        }
    }
}
\State{$^-F[x][S] \gets \discarded$ vertices in $S.V$ except for $x$ }
\Return{$^-F[x][S]$ }
\end{algorithm}

\begin{algorithm}[t]
\SetVline 
\SetFuncSty{textsf}
\SetArgSty{textsf}
\small
\caption{\bf FindAnchoredFollowers($x$, $S$)}
\label{alg:anc_followers}
\Input
{
  $x:$ the anchor candidate, $S:$ a shell component
}
\Output
{
  $^+F[x][S]:$ the anchored follower set of $x$ in $S$
}
\State{$H \gets$ a min heap w.r.t. the layer value of each vertex}
\If{$x \in S.V$}
{
    \State{$x$ is set \survived; $H.push(x)$}
}
\Else
{
    \For{each $u \in N(x, G) \cap S.V$}
    {
        \State{$HS(u) := HS(u) + 1$; $H.push(u)$}
    }
}
\While{$H \not = \emptyset$ }
{
    \State{$u \leftarrow H.pop()$}
    \If{$u \neq x$}
    {
        \State{$d^{+}(u) := HS(u) + |\{v:~v \in N(u, S) \land l(v) \leq l(u) \land (v~is~\survived \vee v \in H)\}| + |\{v:~v \in N(u, S) \land l(v) > l(u) \land v~is~not~discarded\}|$}
        \If {$d^{+}(u) \geq S.c + 1$}
        {
            \State{$u$ is set \survived}

        }
    }
    \If{$u$ is \survived}
    {
        \For{each $v \in N(u, S)$ with $l(v) > l(u)$ and $ v\notin H$}
        {
            \State{$H.push(v)$}
        }
    }
    \Else
    {
        \State{$u$ is set $\discarded$ }
        \State{\textbf{Shrink}($u$) (Algorithm~\ref{alg:shrink})}
    }
}
\State{$^+F[x][S] \gets \survived$ vertices in $S.V$ except for $x$ }
\Return{$^+F[x][S]$ }
\end{algorithm}

\begin{algorithm}[t]
\SetVline 
\SetFuncSty{textsf}
\SetArgSty{textsf}
\small
\caption{\bf Shrink($u$)}
\label{alg:shrink}
\Input
{
   $u:$ the shrinked vertex
}

\For{each \survived neighbor $v$ with $v$ $\not = x$ }
{
 \State{$d^{+}(v) := d^{+}(v) - 1$}
 \If{$d^{+}(v) < S.c + 1$}
 {
   \State{$v$ is set \discarded}
   \State{$T.add(v)$}
 }
}
\For {each $v \in T$ }
{
 \State{\textbf{Shrink}($v$)}
}
\end{algorithm}

Now we introduce our efficient algorithms of computing the collapsed followers and anchored followers of one vertex in one shell component. For each shell component $S$ in $G$, Algorithm~\ref{alg:col_followers} computes $^-F[v][S]$ for each $v \in \mathcal{C}[S]$, and Algorithm~\ref{alg:anc_followers} computes $^+F[v][S]$ for each $v \in \mathcal{A}[S]$. Both of the algorithms need the \textit{higher coreness support} of each $u \in V(G)$, which is introduced as follows.

\vspace{1mm}
\textbf{Higher Coreness Support.} For each $u \in V(G)$, we define its higher coreness support, denoted by $HS(u)$, as the number of $u$'s neighbors having higher coreness than $u$, i.e.,$HS(u) = |\{v \in N(u, G)~:~c(v) > c(u)\}|$. The higher coreness supports are incidentally computed and recorded when traversing each vertex's neighbours in Algorithm~\ref{alg:sdecomp}.

\vspace{1mm}
\textbf{Collapsed Followers Computation.}
In Algorithm~\ref{alg:col_followers}, we utilize a queue $Q$ (Line 1) to explore the collapsed followers, starting from the collapser vertex $x$. If $x \in S.V$, $x$ is set \textit{discarded} and pushed into $Q$ (Line 2-3). Note that, all the vertices in $S.V$ are not \textit{discarded} initially, and any \textit{discarded} vertex (except for $x$) becomes a collapsed follower. If $x \notin S.V$, for each $u \in N(x, G) \cap S.V$, we reduce $HS(u)$ by 1 and push $u$ into $Q$ (Line 4-6). Then we traverse $Q$ until it becomes empty (Line 7). Each time when we pop the top vertex $u$ (Line 8), if $u \neq x$, we need to decide whether it is \textit{discarded} (Line 9-12). Specifically, we compute a degree upper bound $d^+(u)$ as Line 10 shows. If $d^+(u) < S.c$, $u$ is set \textit{discarded}, and we push each $v \in N(u, S)$ into $Q$ (Line 13-15). Note that, we can avoid repeatedly push the same vertex into $Q$ (Line 14). After traversing $Q$, all the \textit{discarded} vertices in $S.V$ except for $x$ form $^-F[x][S]$ (Line 16).

\vspace{1mm}
\textbf{Anchored Followers Computation}.
We use Algorithm~\ref{alg:anc_followers} to compute $^+F[x][S]$, the anchored followers of $x$ in $S$. The algorithm is straightforward adapted from the Algorithm 4 of~\cite{SigmodLinghu}, so we omit the theoretical analysis and the proof of algorithm correctness. Please also refer to \textit{Degree Check} and Theorem 4.15 in section 4.4 of~\cite{SigmodLinghu}. Instead, Example~\ref{ex:followers} illustrates the process of Algorithm~\ref{alg:anc_followers} in detail. The core of Algorithm~\ref{alg:anc_followers} is utilizing the \textit{layer value} of each vertex, which needs to be illustrated firstly as follows.

\vspace{1mm}
\textbf{Layer Value}.
The vertices in the \ks can be further divided to different vertex sets, named layers, according to their deletion sequence in the core decomposition (Algorithm~\ref{alg:coredecomp}).
We use $H_{k}^{i}$ to denote the $i$-layer of the $k$-shell, which is the set of vertices that are deleted in the $i$-th batch.
Specifically, when $i = 1$, $H_{k}^{i}$ is defined as $\{u:~deg(u,C_k(G))<k+1~\wedge~u \in C_k(G)\}$.
The deletion of the 1st-layer will produce the 2nd-layer.
Recursively, when $i > 1$, $H_{k}^{i} = \{u:~deg(u, G_i)<k+1~\wedge~u \in G_i\}$ where $G_1 = C_k(G)$ and $G_i$ is the subgraph induced by $V(G_{i-1})\setminus H_k^{i-1}$ on $C_k(G)$.
For each vertex $u \in V(G)$, there is only one $H_{k}^{i}$ s.t. $u \in H_{k}^{i}$ during core decomposition, then we denote $l(u) = i$ as the \textit{layer value} of $u$. Obviously, the layer values are incidentally computed and recorded when conducting Algorithm~\ref{alg:coredecomp}.

\begin{example}
\label{ex:followers}
For the graph $G$ in Figure~\ref{fig:shell_cp} (Table~\ref{tb:shell_component}), we follow all the notations in Example~\ref{ex:ks}. For $v_b$, we need to compute $^+F[v_b][S_2]$, $^+F[v_b][S_3]$ and $^+F[v_b][S_4]$. We follow the process of \texttt{FindAnchoredFollowers}($v_b, S_3$) in Algorithm~\ref{alg:anc_followers} as an example. As $v_b \notin S_3.V$ (Line 4), we have $HS(v_d) = 2 + 1 = 3$ (originally $HS(v_d) = 2$ w.r.t. $v_k$ and $v_l$) and $v_d$ is pushed into $H$ (Line 5-6). After $v_d$ is popped (Line 8), we compute $d^+(v_d)$ (Line 10). Because $l(v_d) = 1$, $l(v_c)$ = 2 and $v_c$ is not \textit{discarded}, $d^+(v_d) = 3 + 0 + 1 = 4$. Since $S_3.c = 3$, $v_d$ is set \textit{survived} (Line 11-12) and $v_c$ is pushed into $H$ (Line 13-15). Then $v_c$ is popped and $d^+(v_c) = 3 + 1 + 0 = 4$, so $v_c$ is set \textit{survived} and no more vertex is pushed into $H$. Now we can return $^+F[v_b][S_3] = \{v_c, v_d\}$. Please note that, any vertex is neither \textit{discarded} or \textit{survived} unless it is explicitly set so. Once a vertex needs to be set \discarded (Line 17), it may cause a cascade of vertices becoming \discarded. We call Algorithm~\ref{alg:shrink} (Line 18) which recursively discards all the need-to-be vertices.
\end{example}

%% file: experiments.tex
\section{Experimental Evaluation}
\label{sec:exp}

\begin{table}
\small
  \centering
  \caption{Statistics of Datasets}
    \begin{tabular}{|l|l|l|l|l|l|}
      \hline
      \textbf{Dataset}  & \textbf{Nodes}  & \textbf{Edges} & $d_{avg}$ & $d_{max}$ & $k_{max}$ \\ \hline \hline

      \texttt{Facebook(\textbf{F.})}  &  22,470 & 170,823 & 15.2 & 709 & 56 \\ \hline
      \texttt{Brightkite(\textbf{B.})}  &  58,228 & 194,090 & 6.7 & 1098 & 52 \\ \hline
      \texttt{Github(\textbf{H.})} &  37,700 & 289,003 & 15.3 & 9458 & 34 \\ \hline
      \texttt{Gowalla(\textbf{G.})}  & 196,591 & 456,830 & 9.2 & 10721 & 51 \\ \hline
      \texttt{NotreDame(\textbf{N.})}  & 325,729 & 1,497,134 & 6.5 & 3812 & 155 \\ \hline
      \texttt{Stanford(\textbf{S.})}  & 281,903 & 2,312,497 & 16.4 & 38626 & 71 \\ \hline
      \texttt{Youtube(\textbf{Y.})}  & 1,134,890 & 2,987,624 & 5.3 & 28754 & 51 \\ \hline
      \texttt{DBLP(\textbf{D.})} &  1,566,919 & 6,461,300 & 8.3 & 2023 & 118 \\ \hline
    \end{tabular}
\label{tb:datasets}
\end{table}

\begin{table}
\small
  \centering
  \caption{Percentage of Valid Collapsers \& Anchors}
    \begin{tabular}{|l|l|l|l|l|l|l|l|l|}
      \hline
      \textbf{Dataset} & \texttt{F.} & \texttt{B.} & \texttt{H.} & \texttt{G.} & \texttt{N.} & \texttt{S.} & \texttt{Y.} & \texttt{D.} \\ \hline
      Collapsers & 61\% & 44\% & 49\% & 50\% & 24\% & 38\% & 28\% & 69\% \\ \hline
      Anchors & 69\% & 70\% & 78\% & 74\% & 50\% & 33\% & 64\% & 54\% \\ \hline
    \end{tabular}
\label{tb:valid_percentage}
\end{table}

\begin{figure}[t]
\begin{center}
\includegraphics[width=0.7\columnwidth]{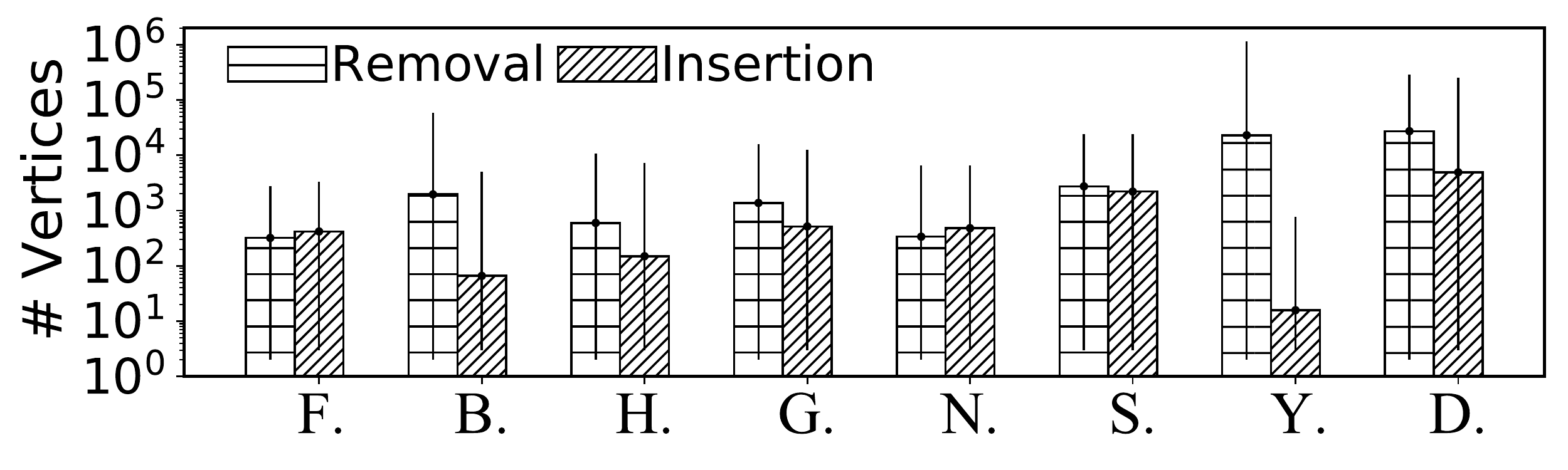}
\end{center}
\caption{\# Updated Vertices w.r.t. Edge Streaming}
\label{fig:updated_vertices}
\end{figure}

\begin{figure*}[htp]
\begin{center}
    \subfigure[Facebook]{
    \includegraphics[width=0.24\columnwidth]{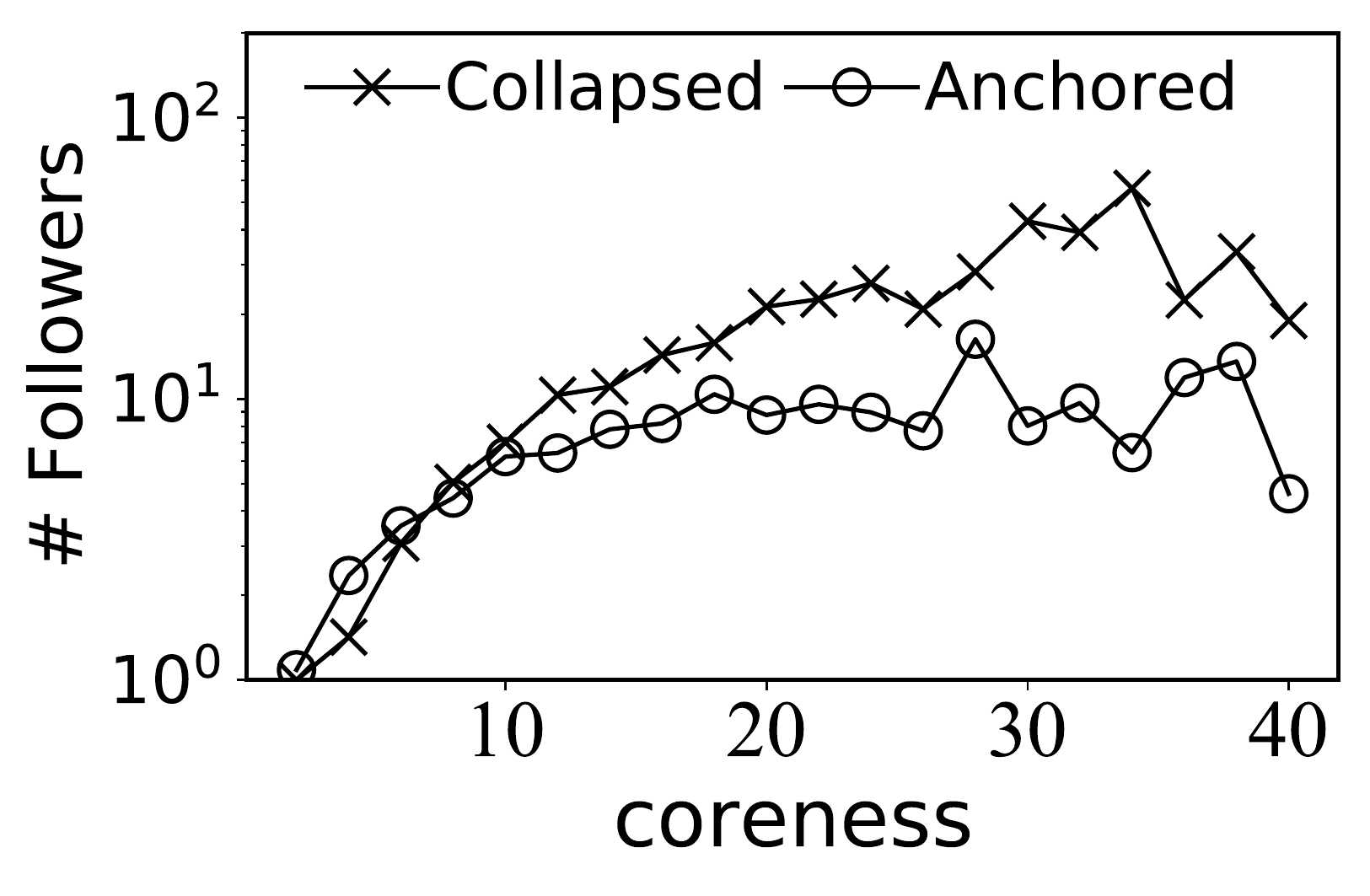}}
    \subfigure[Brightkite]{
    \includegraphics[width=0.24\columnwidth]{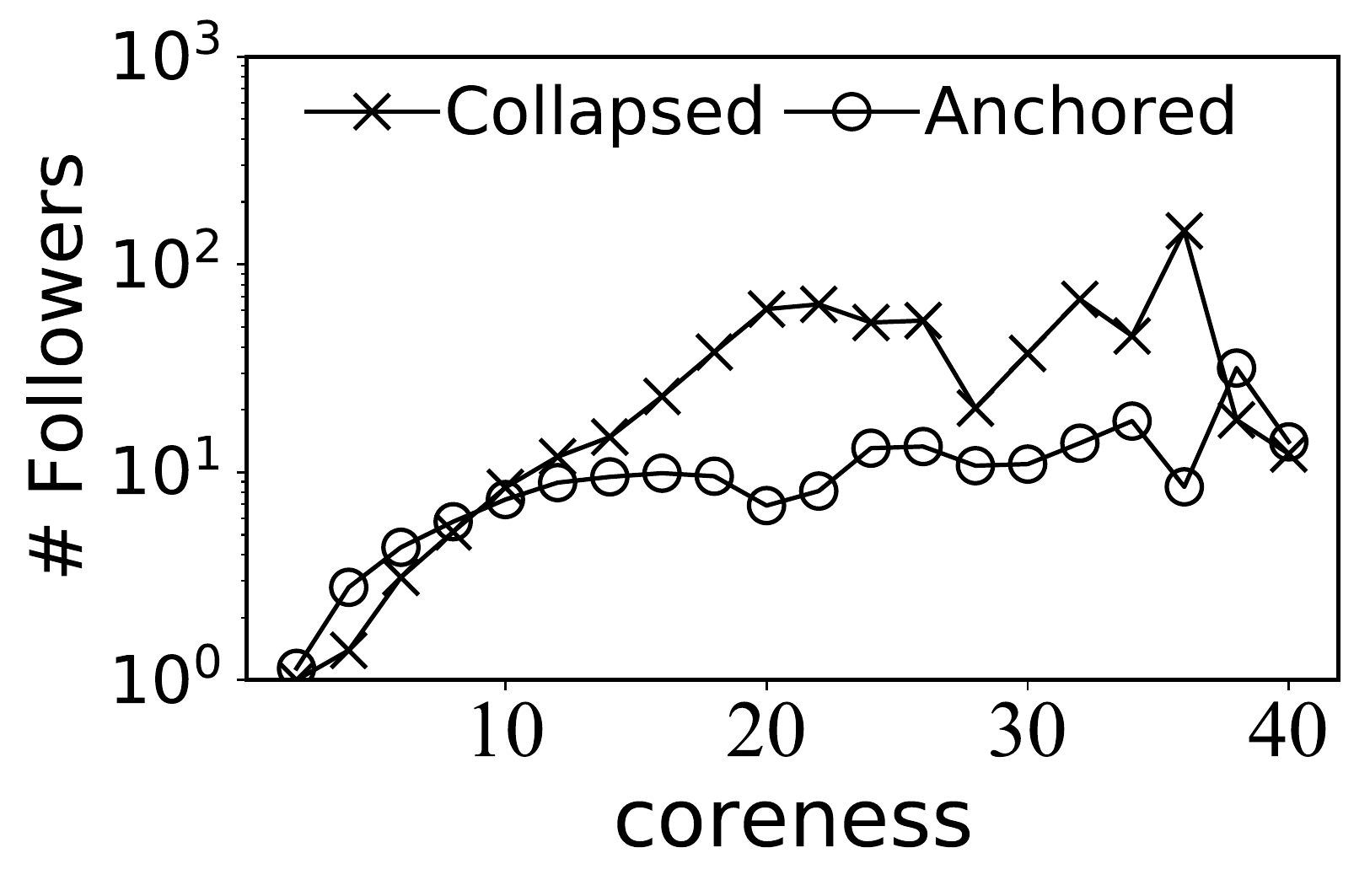}}
    \subfigure[Github]{
    \includegraphics[width=0.24\columnwidth]{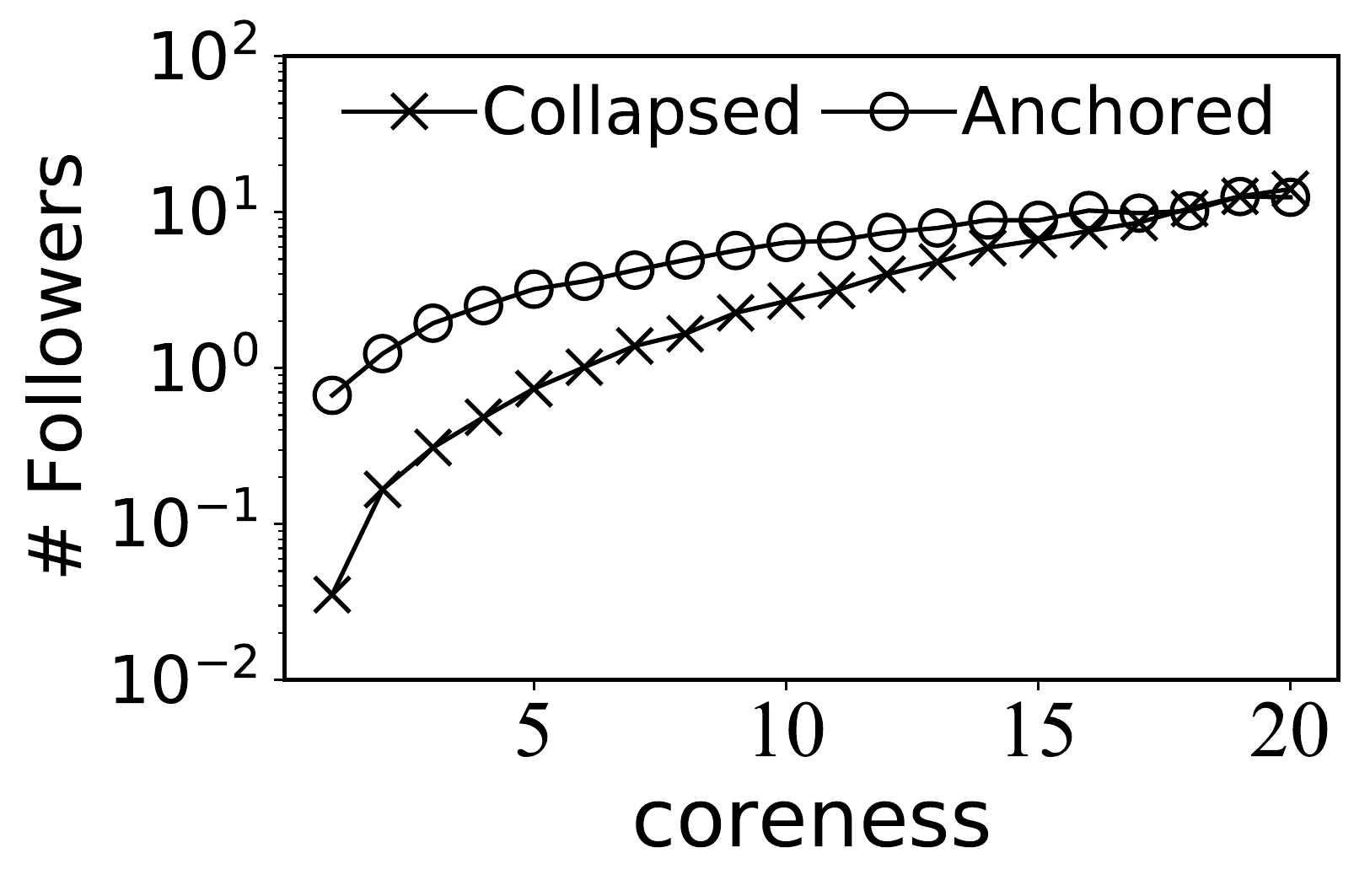}}
    \subfigure[Gowalla]{
    \includegraphics[width=0.24\columnwidth]{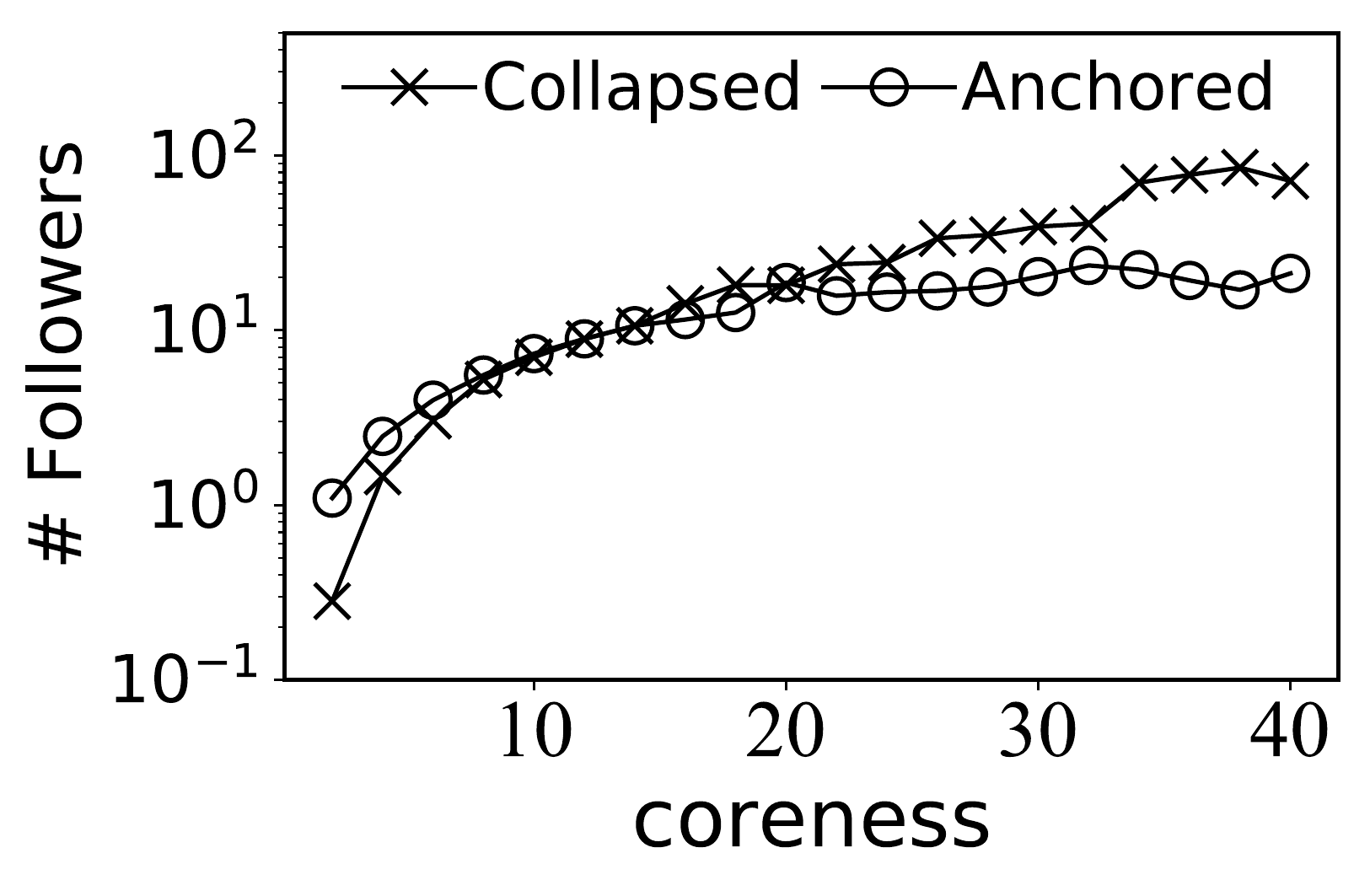}}
    \vspace{-2mm}
    \subfigure[NotreDame]{
    \includegraphics[width=0.24\columnwidth]{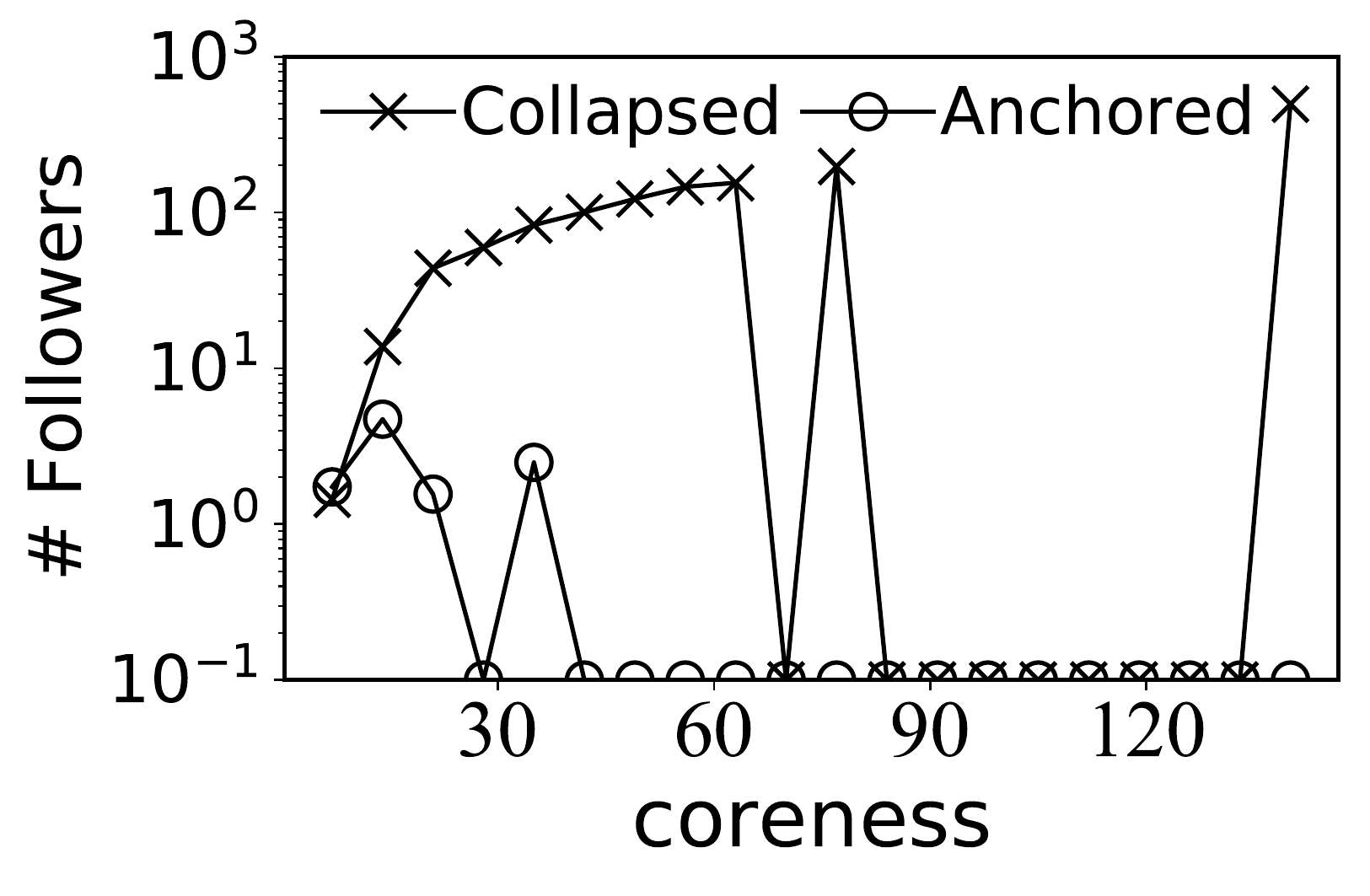}}
    \subfigure[Stanford]{
    \includegraphics[width=0.24\columnwidth]{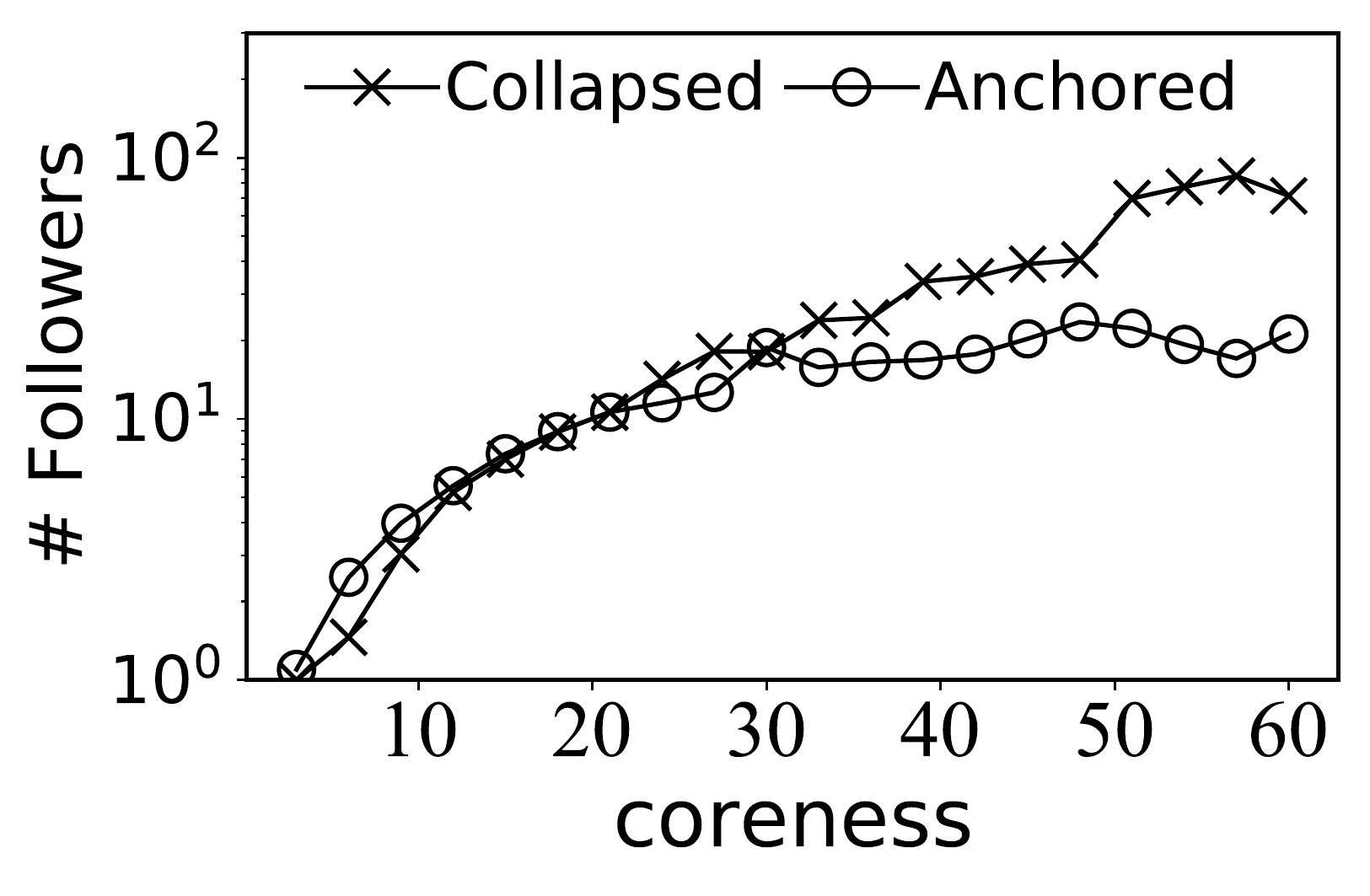}}
    \subfigure[Youtube]{
    \includegraphics[width=0.24\columnwidth]{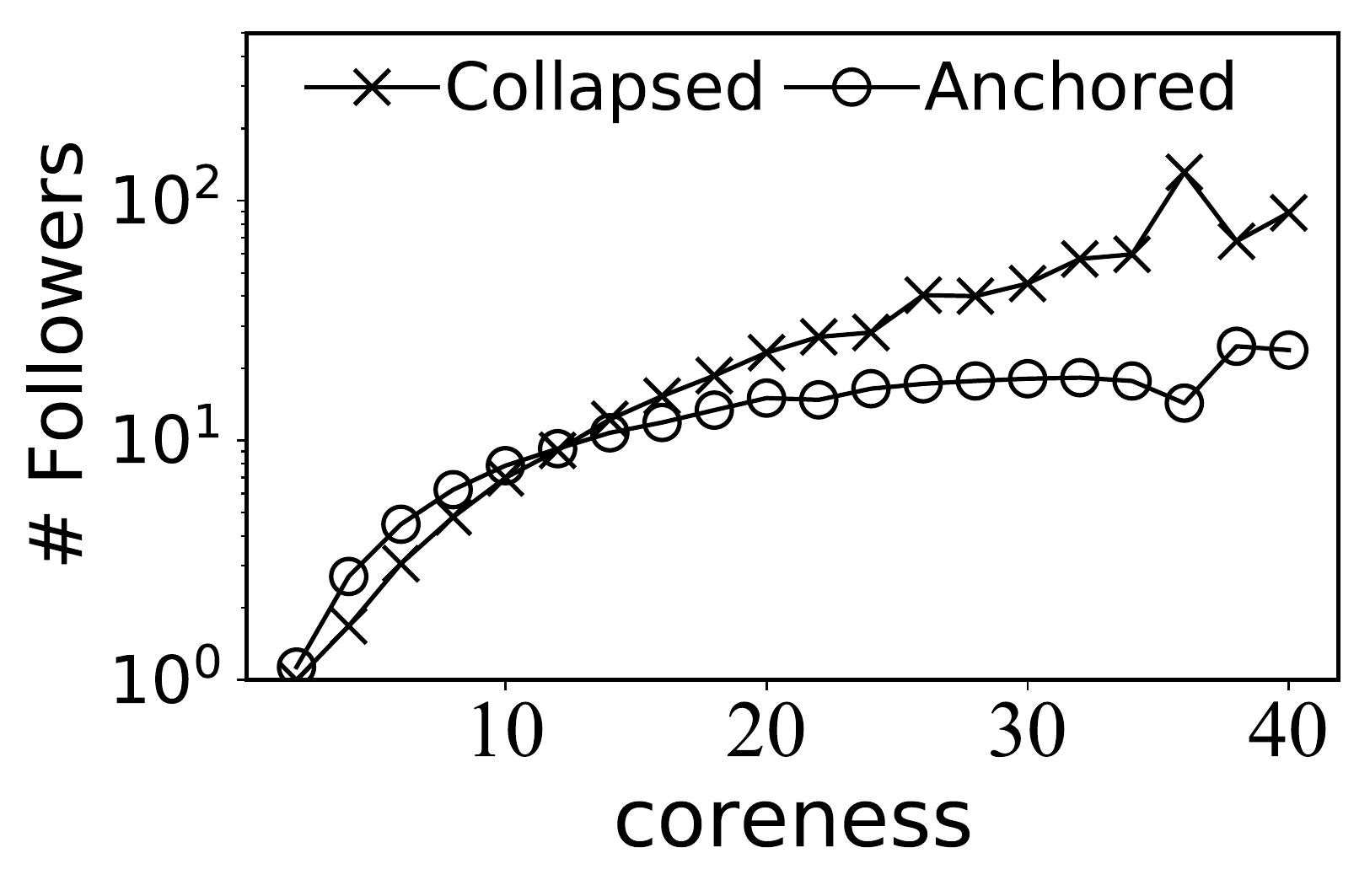}}
    \subfigure[DBLP]{
    \includegraphics[width=0.24\columnwidth]{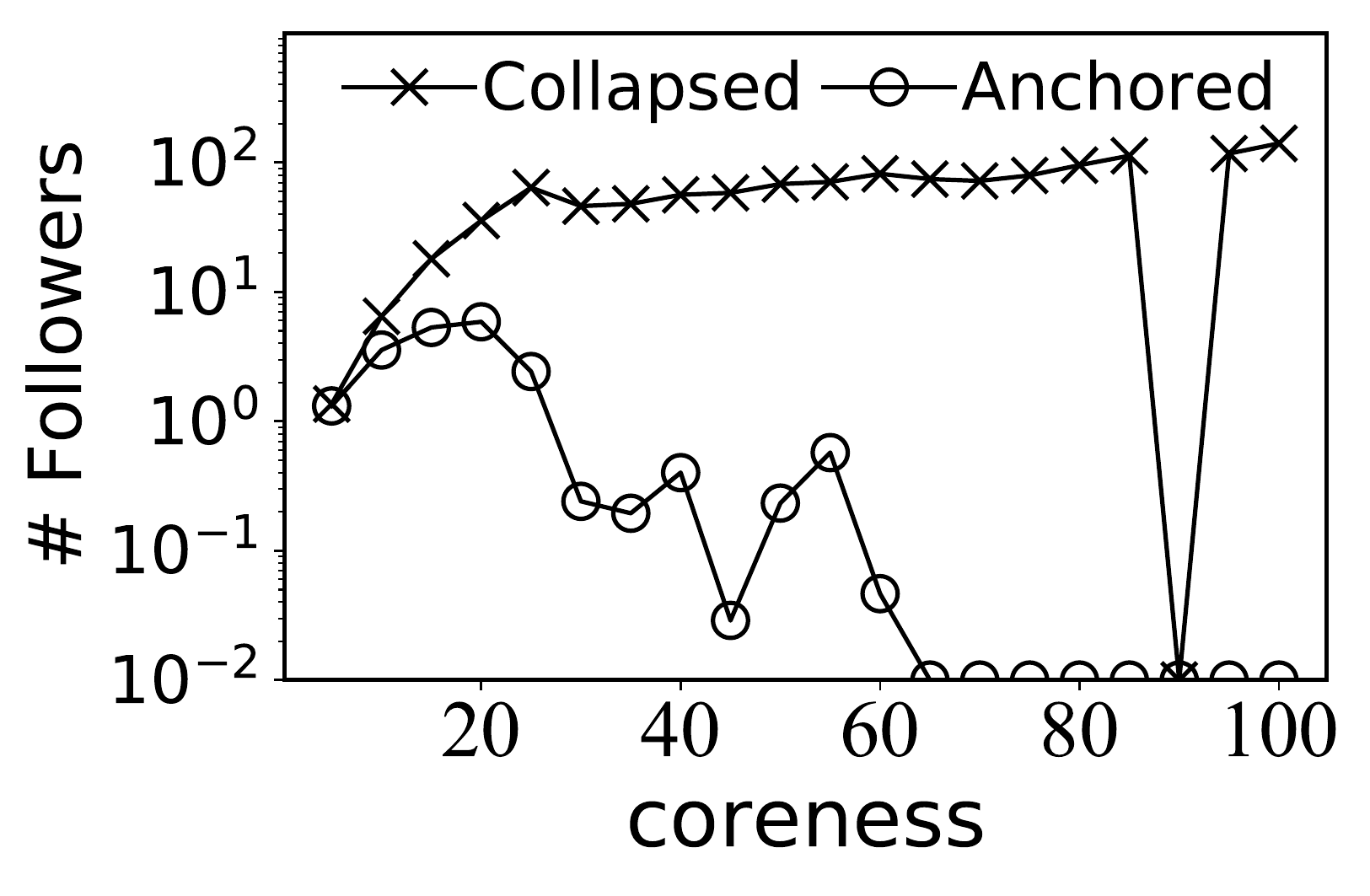}}
\end{center}
\caption{\# Followers Distribution on Vertex Coreness}
\label{fig:coreness_distr}
\end{figure*}

\begin{figure*}[htp]
\begin{center}
    \subfigure[Facebook]{
    \includegraphics[width=0.24\columnwidth]{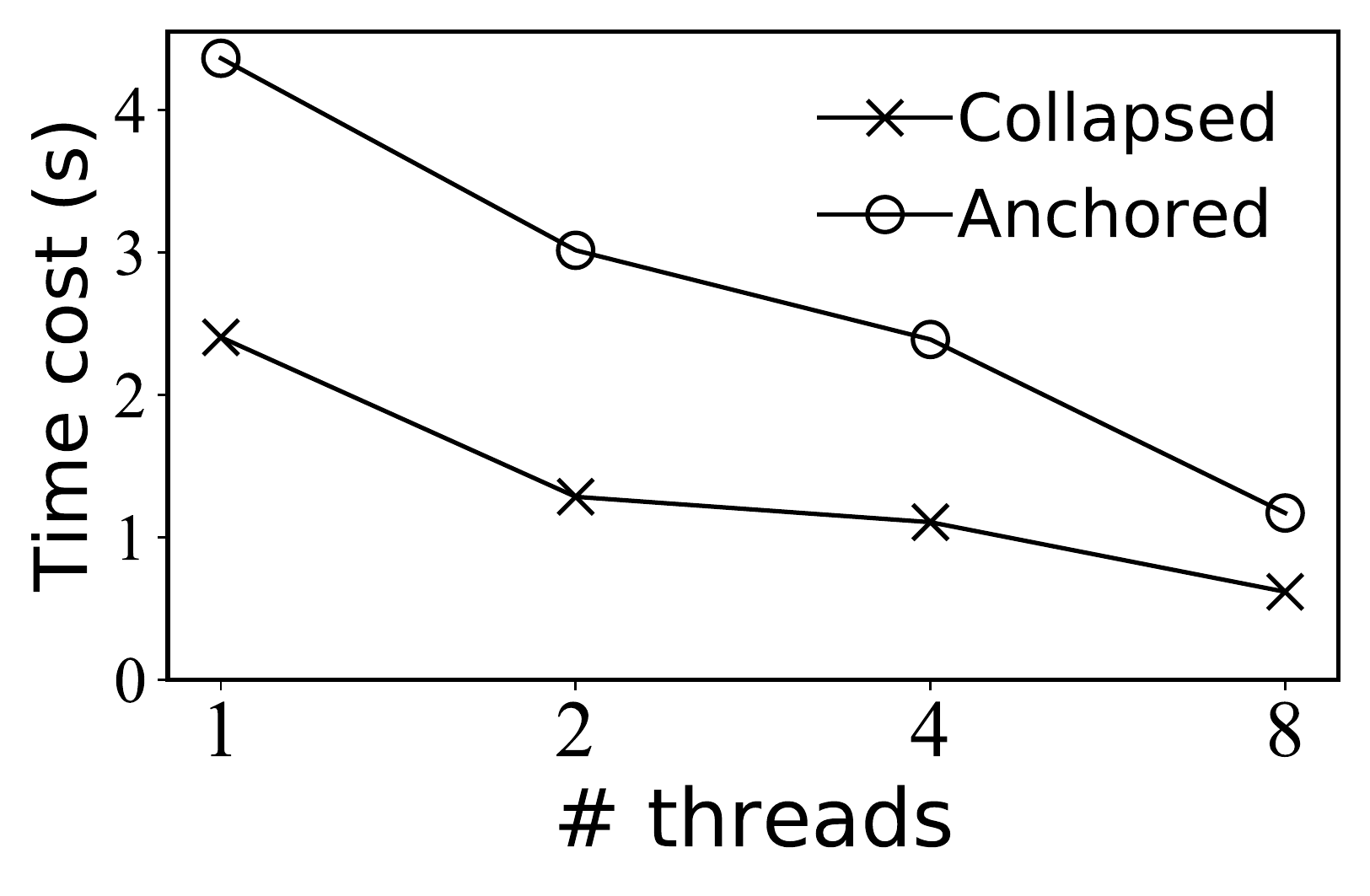}}
    \subfigure[Brightkite]{
    \includegraphics[width=0.24\columnwidth]{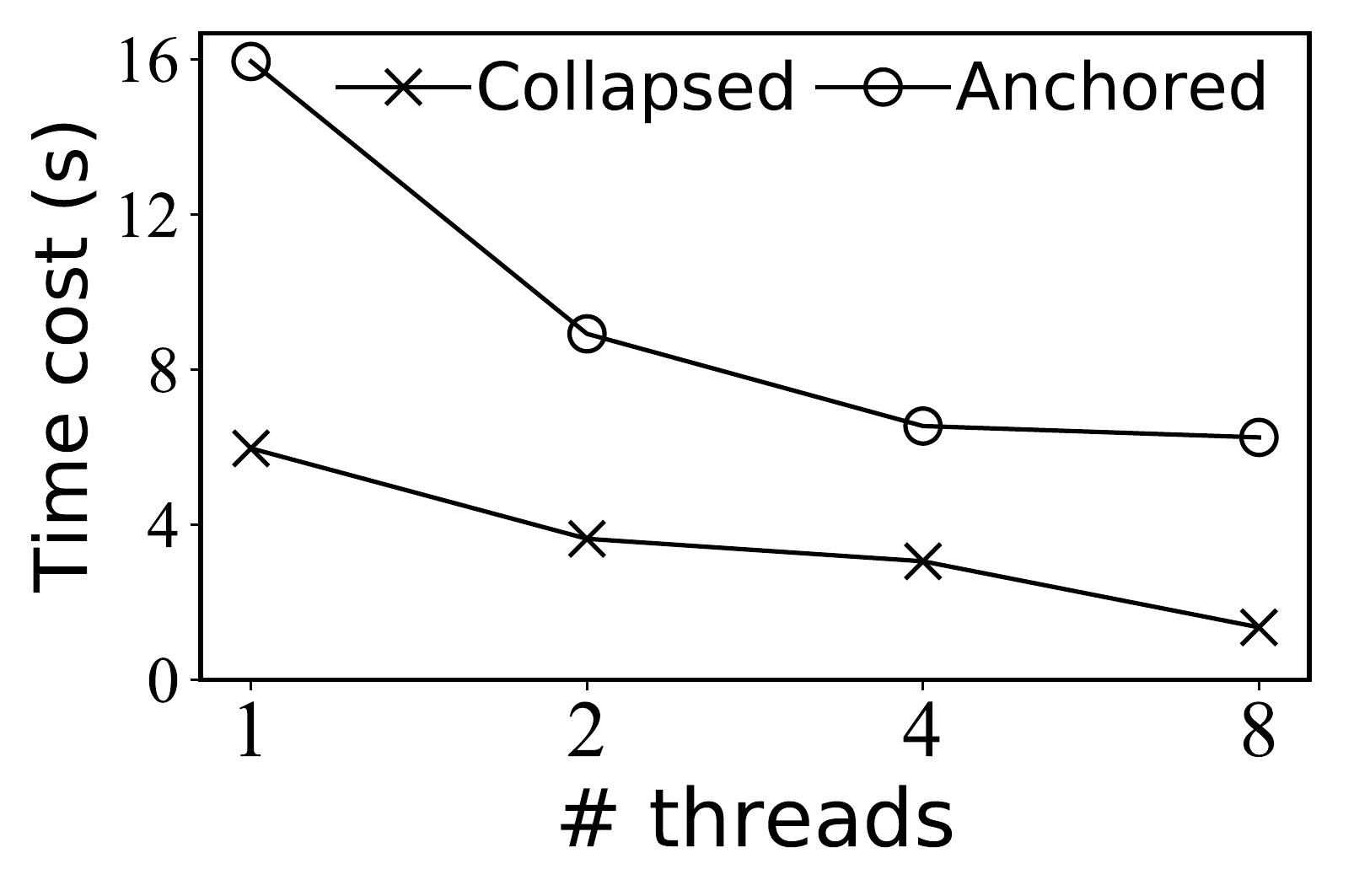}}
    \subfigure[Github]{
    \includegraphics[width=0.24\columnwidth]{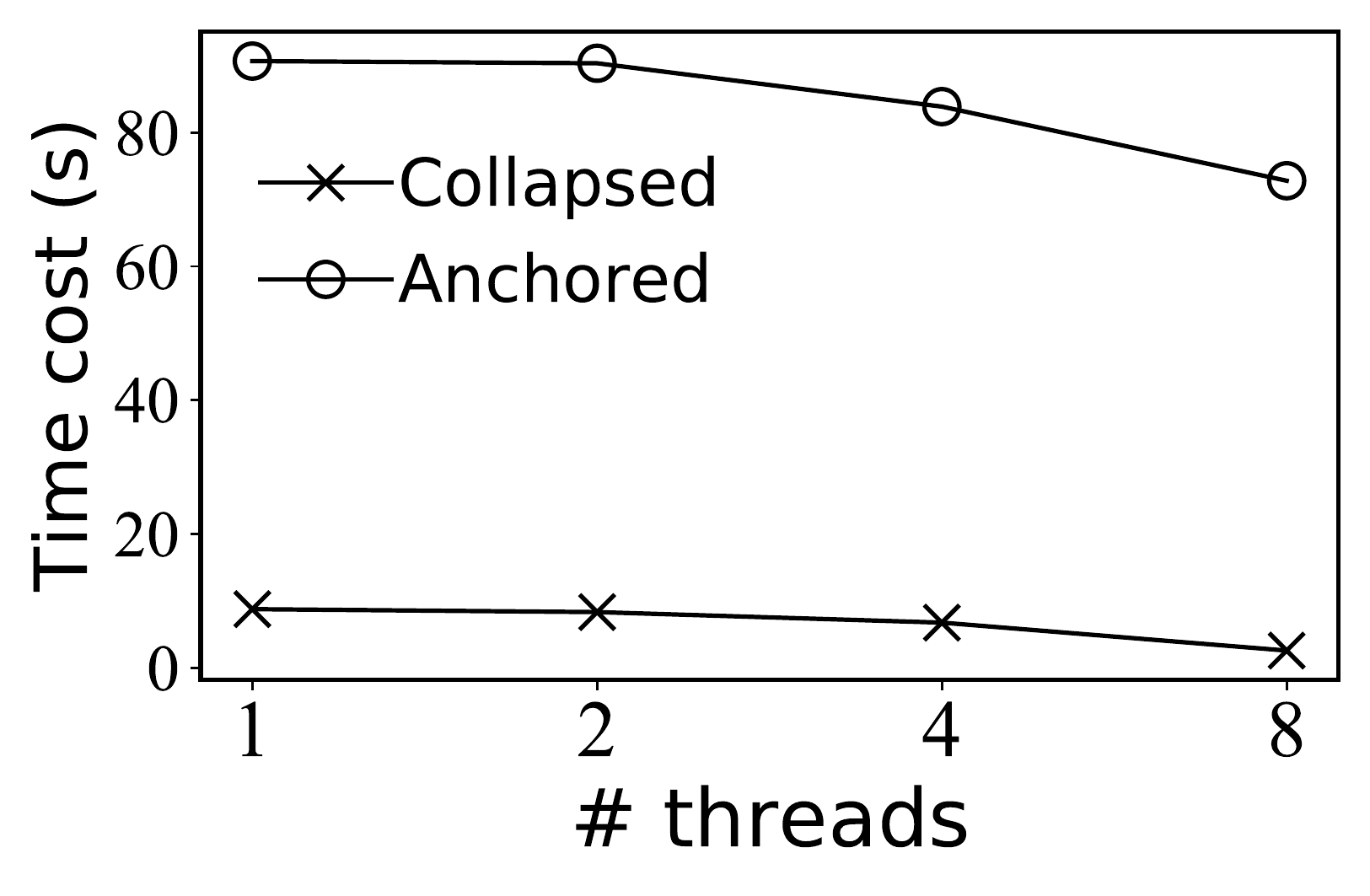}}
    \subfigure[Gowalla]{
    \includegraphics[width=0.24\columnwidth]{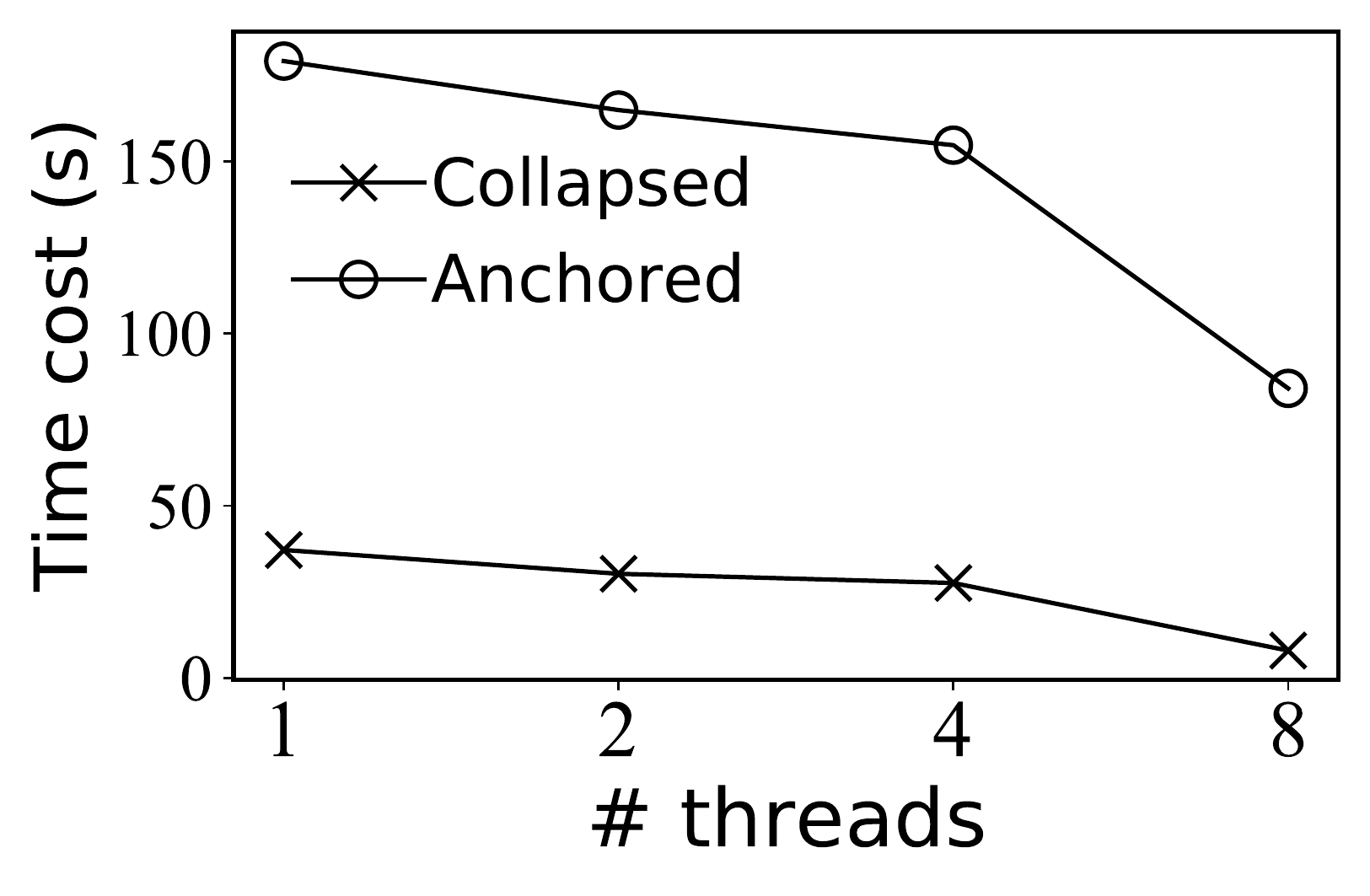}}
    \vspace{-2mm}
    \subfigure[NotreDame]{
    \includegraphics[width=0.24\columnwidth]{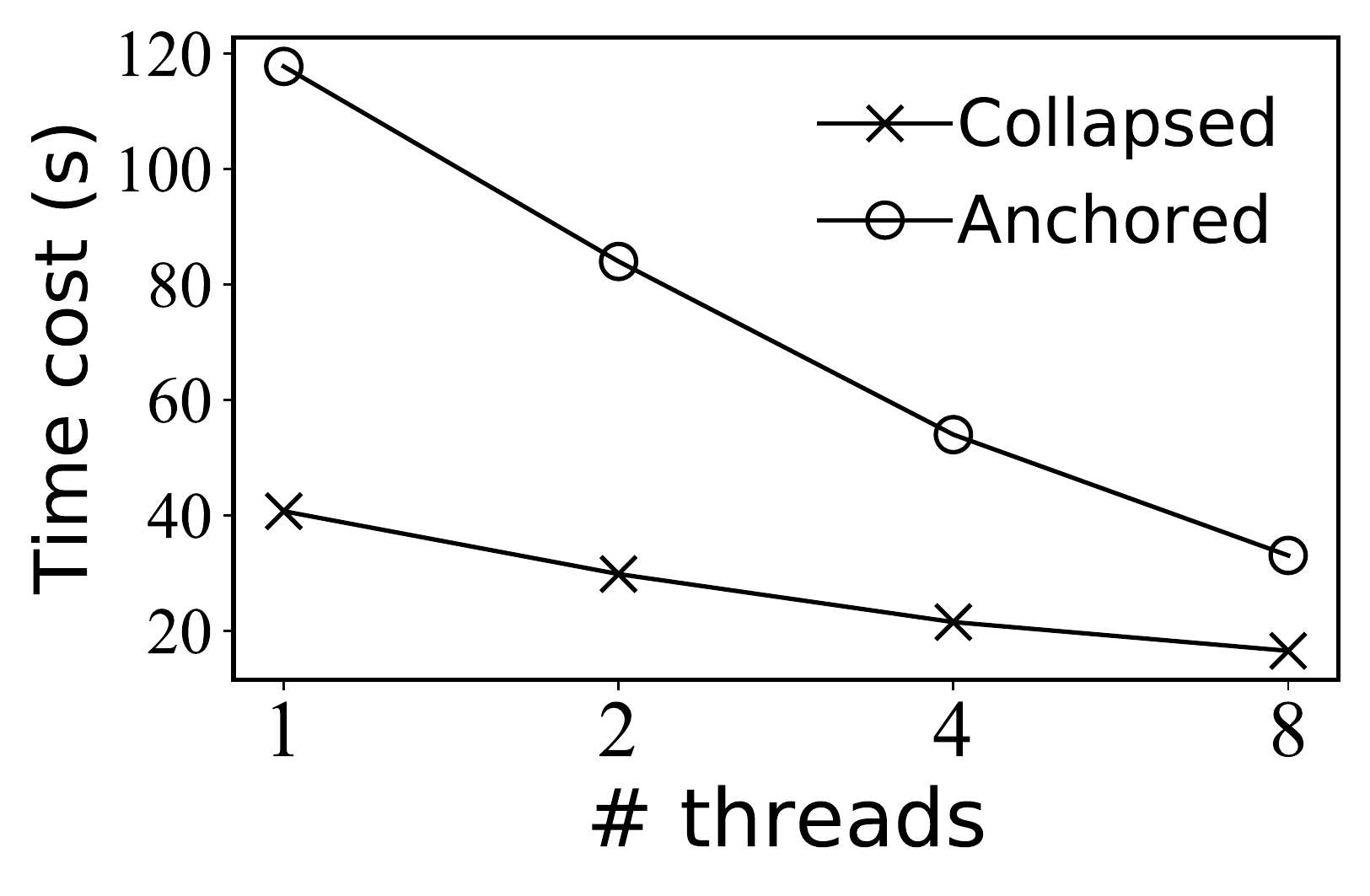}}
    \subfigure[Stanford]{
    \includegraphics[width=0.24\columnwidth]{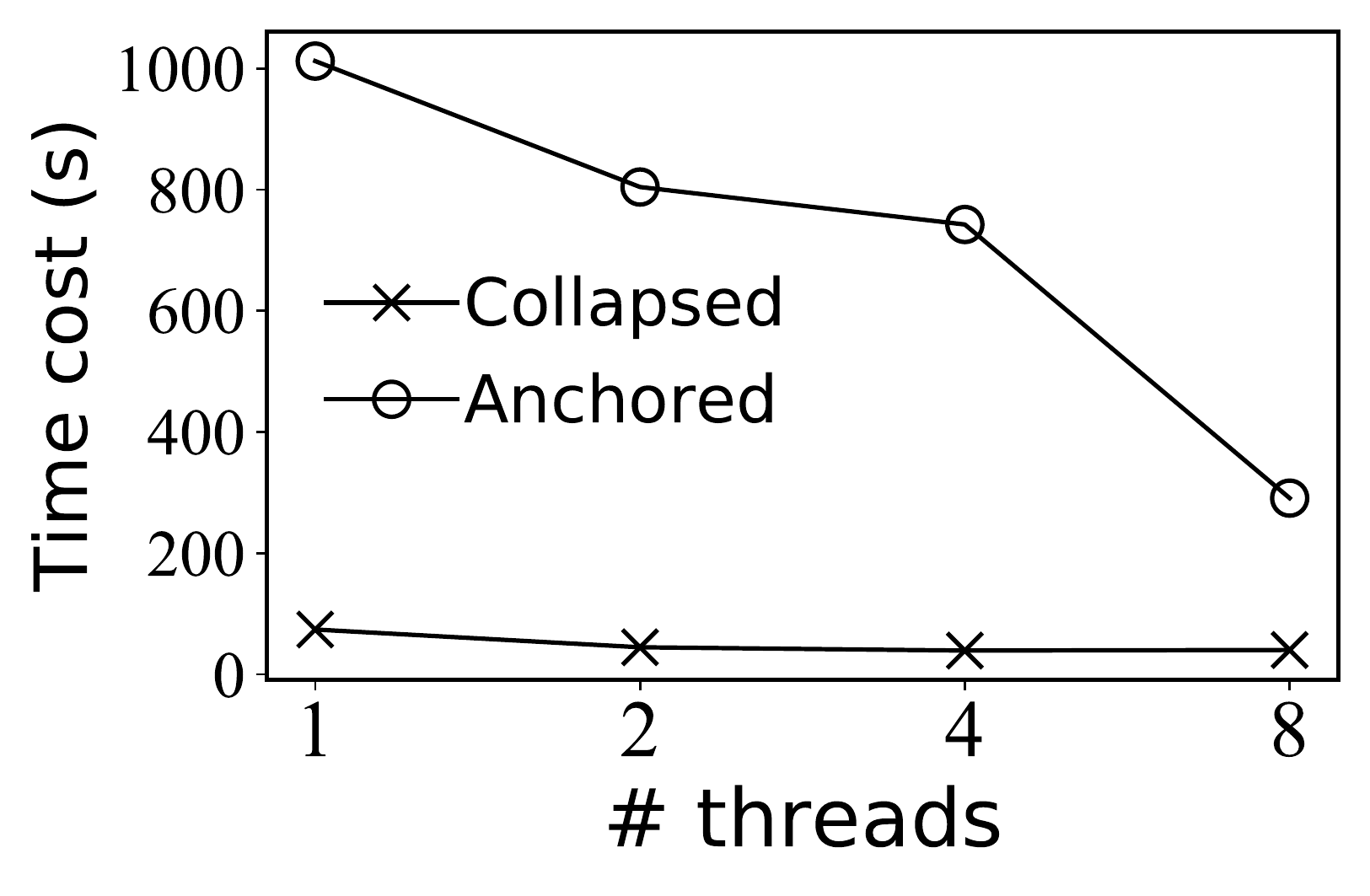}}
    \subfigure[Youtube]{
    \includegraphics[width=0.24\columnwidth]{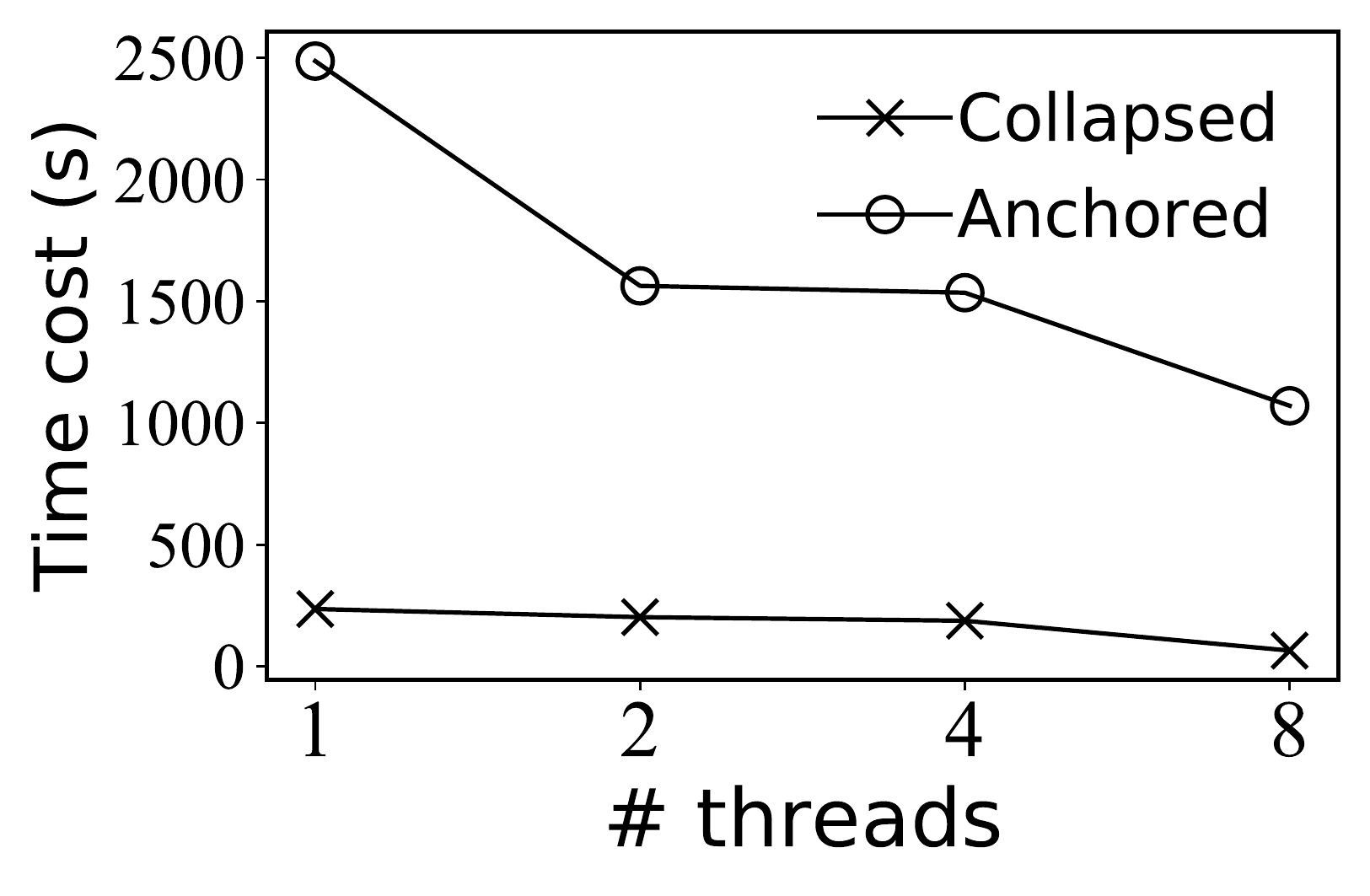}}
    \subfigure[DBLP]{
    \includegraphics[width=0.24\columnwidth]{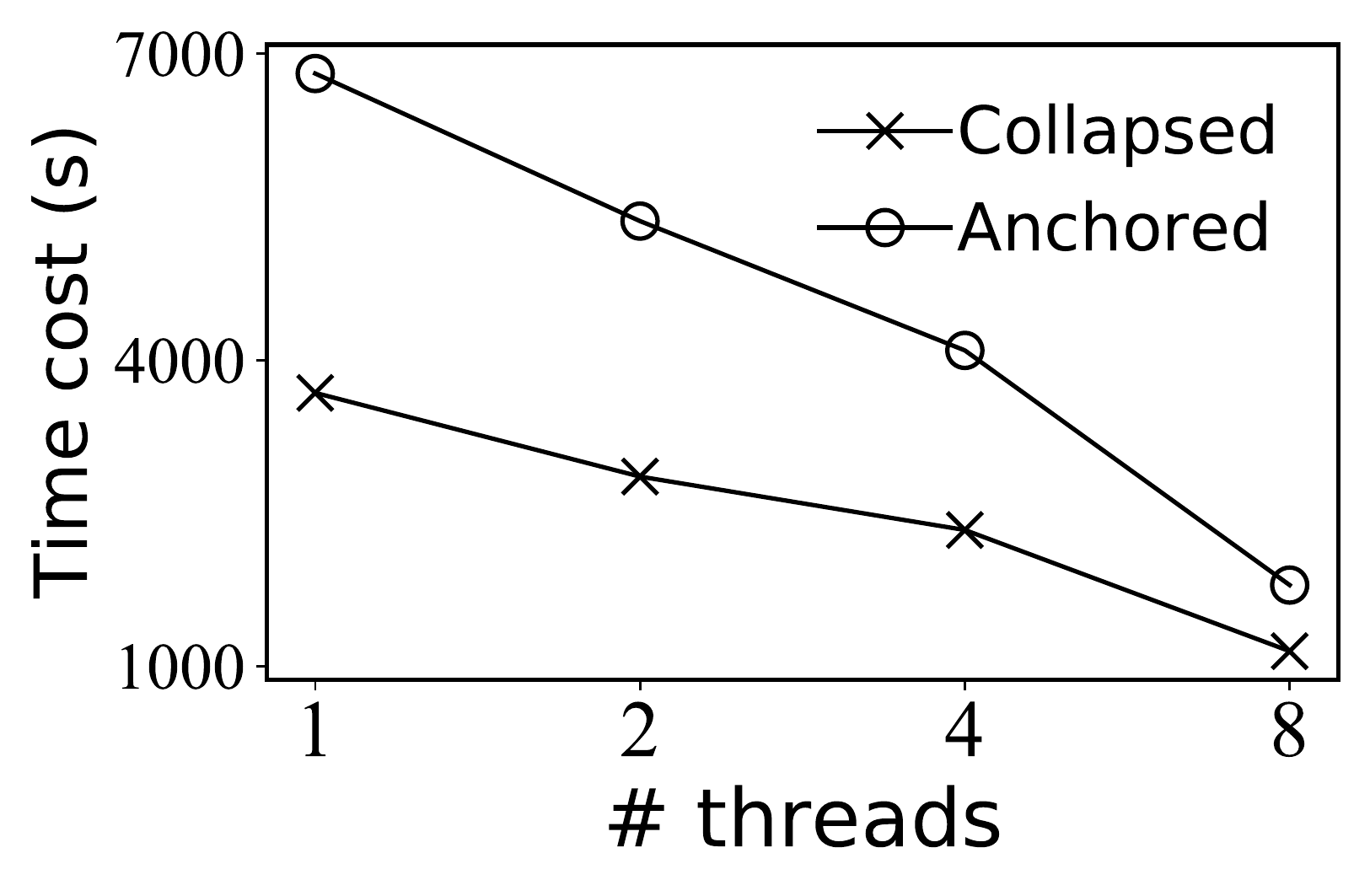}}
\end{center}
\caption{Computation Efficiency and Scalability}
\label{fig:efficiency_scalability}
\end{figure*}

\vspace{1mm}
\noindent \textbf{Datasets.} We use 8 real-life datasets for experiments. \texttt{Facebook}, \texttt{Brightkite}, \texttt{Github}, \texttt{Gowalla} \& \texttt{Youtube} are from \cite{snap}. \texttt{NotreDame}, \texttt{Stanford} and \texttt{DBLP} are from \cite{konect}. Due to the Space limitation, we abbreviate each dataset's name as a unique bold capital letter when necessary as in Table~\ref{tb:datasets}. Table~\ref{tb:datasets} also shows the statistics of the datasets, listed in increasing order of edge numbers.

\vspace{1mm}
\noindent \textbf{Parameters.} All the programs are implemented in C++ and compiled with G++ on Linux. The server has 3.4GHz Intel Xeon CPU with 4 cores (8 threads available) and Redhat system. We adopt \texttt{OpenMP} to utilize the multithreads of the machine.

\subsection{Effectiveness}

\vspace{1mm}
\noindent \textbf{Percentage of Valid  Collapsers \& Anchors.}
In Table~\ref{tb:valid_percentage}, for each dataset, we present the percentage of vertices which have non-empty collapsed follower set (resp. anchored follower set), which we call the valid collapsers (resp. valid anchors). Firstly, we can find the percentages are higher than 50\% on most datasets for both the valid collapsers and valid anchors. This means both the collapsing and anchoring behaviors of vertices indeed have apparent effects on the network, which demonstrates the necessity of monitoring each vertex's collapsed and anchored follower set. Secondly, we find on some datasets such as \texttt{Stanford} and \texttt{DBLP}, the number of valid collapsers are more than valid anchors. However, on all the other datasets, the number of valid anchors are more than valid collapsers. This means vertex collapsing and vertex anchoring have different influence on different networks due to the different network structures, thus it is necessary to monitor each vertex's influence via both the collapsed and anchored follower set.

\vspace{1mm}
\noindent \textbf{Followers Distribution on Vertex Coreness.}
In Figure~\ref{fig:coreness_distr}, we present the distribution of the number of followers (collapsed and anchored) on vertex's coreness value. For each dataset, $k_{max}$ denotes the largest coreness value from all the vertices, which is shown in the statistics of Table~\ref{tb:datasets}. We divide the coreness range $[1, k_{max}]$ into 20 equal integer-width intervals, with the last remained interval combined with the second last interval. For each coreness value $x_i$ on the horizontal ordinate of Figure~\ref{fig:coreness_distr}, we compute the mean number of collapsed followers (resp. anchored followers) of all the vertices having their own coreness value within $(x_{i-1}, x_i]$. Firstly, we can find on most coreness values on most datasets, the mean number of collapsed followers are more than anchored followers. However, remember that in Table~\ref{tb:valid_percentage}, the number of valid anchors are more than valid collapsers on the contrary. This means the collapsed followers and anchored followers reflect two different views of node influence on network structural stability and it is necessary to monitor both of them. Secondly, we find that, except when the coreness value is little, the number of collapsed and anchored followers are not always positive-correlated with the coreness value. This demonstrates that, we cannot simply decide a node's influence on network structural stability simply based on its own engagement, and it is essential to actually compute the collapsed and anchored followers.

\vspace{1mm}
\noindent \textbf{Amount of Updated Vertices w.r.t. Edge Streaming.}
On each dataset, we randomly remove 100 edges and insert 100 edges. Each time an edge is removed or inserted, we record the number of vertices which have updates in either the collpased or anchored follower set. The result is shown in Figure~\ref{fig:updated_vertices}. The main boxes show the mean amount of updated vertices. We can find that, either removing or inserting a single edge can cause the update amount of other vertices from $10^1$ to $10^4$ on average. We also add the error-bar on each box to show the minimum and maximum amount of updated vertices among the 100 edge removal and insertion. We can find that the minimum amount of updated vertices are generally close to 1 and the maximum amount can reach up to around $10^5$ in extreme cases (\texttt{Brightkite} and \texttt{Youtube}). The significant amount of vertices in need of updating the follower sets demonstrates the necessity of our maintaining technique regarding edge streaming.

\vspace{1mm}
\subsection{Efficiency}

\vspace{1mm}
\noindent \textbf{Offline Computation Efficiency and Scalability.}
We vary the number of threads from 1, 2, 4 to 8 and test the efficiency and scalability of our offline algorithm of computing the collapsed and anchored follower set of each vertex. The time cost is presented in Figure~\ref{fig:efficiency_scalability}. We can find that, on all the datasets, the time cost tends to be less with increasing the number of threads, for the computation of both collapsed and anchored followers. This means our parallel algorithm based on separating the computation into each shell component helps us take the advantage of parallel architechure.

\vspace{1mm}
\noindent \textbf{Online Maintenance Efficieny.}
We test the efficiency of our online maintenance algorithm. In all the datasets, we randomly sample 100 edges for removal and insertion respectively. Each time an edge is removed (resp. inserted), we record the time cost of maintaining both the collapsed and anchored follower set for all the vertices. In Figure~\ref{fig:maintenance_time}, the main boxes compare the mean time cost of the 100 maintenance to the offline computation time, for edge removal (a) and edge insertion (b), respectively. We can find that, the maintenance time is 2 to 4 orders of magnitude faster than the offline computation, and the maintenance of edge insertion is generally faster than the maintenance of edge removal. We also add the error-bars to show the minimum and maximum maintenance time among the 100 edge removals and 100 edge insertions. We can find that, in the worst cases, the maintenance time cost is still less than the time cost of offline computation. And in the optimal cases, the maintenance time cost is constantly little regardless of the scale of datasets. Overall, the experiments presented in Figure~\ref{fig:maintenance_time} demonstrate the efficiency of our maintaining technique for evolving networks.

\begin{figure}[htp]
\begin{center}
    \subfigure[Edge Removal]{
    \includegraphics[width=0.7\columnwidth]{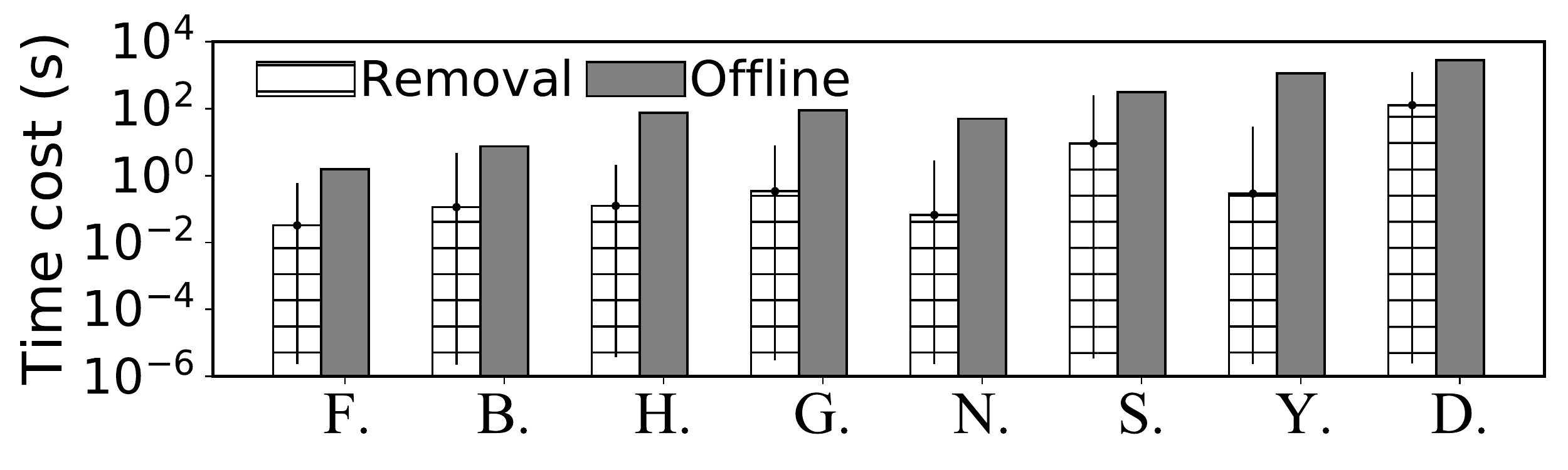}}
    \subfigure[Edge Insertion]{
    \includegraphics[width=0.7\columnwidth]{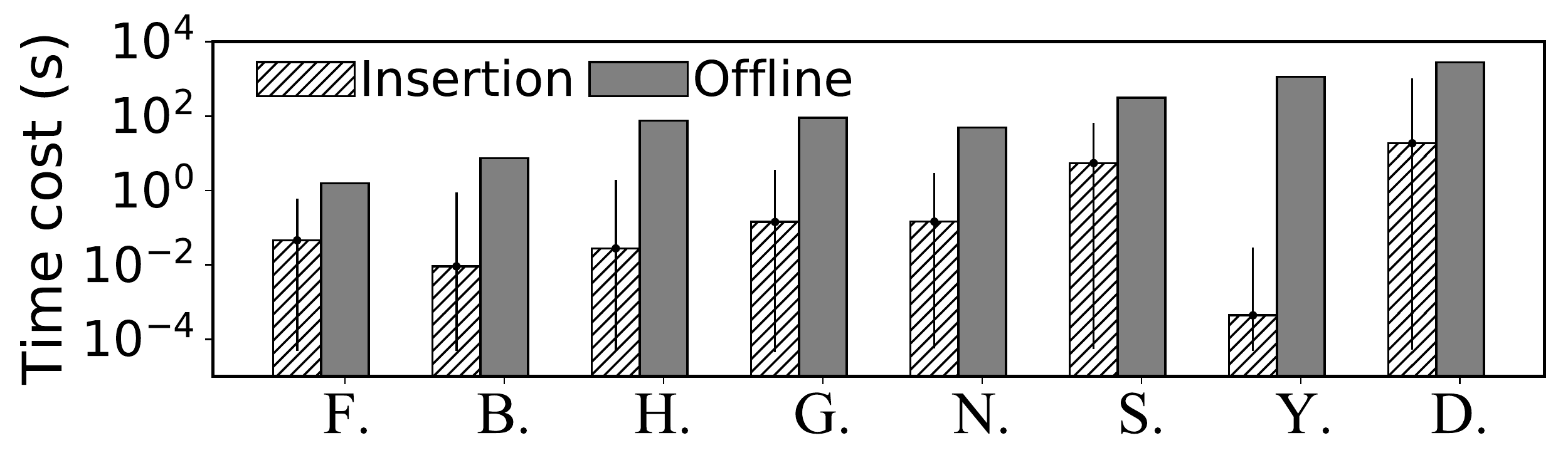}}
\end{center}
\caption{Maintenance Time w.r.t. Edge Streaming}
\label{fig:maintenance_time}
\end{figure}